\documentclass{IEEEtran}

\usepackage{balance}
\usepackage{amsmath,graphicx}
\usepackage{amsfonts}
\usepackage{color}
\usepackage{amssymb}
\usepackage{amsthm}
\usepackage{bm}
\usepackage{cite}

\newtheorem{problem}{Problem}
\newtheorem{lemma}{Lemma}

\newtheorem{remark}{Remark}
\usepackage{cuted}
\usepackage{array,multirow}
\usepackage{stfloats}
\usepackage{bbm}
\usepackage{mathtools}
\usepackage{algorithm,algorithmic} %
\usepackage{multicol}
\usepackage{booktabs}
\usepackage{multirow}
\usepackage{graphicx}
\usepackage{framed}
\newcommand{\be}{\begin{equation}}
\newcommand{\ee}{\end{equation}}
\newcommand{\ist}{\hspace*{.4mm}}
\newcommand{\rmv}{\hspace*{-.4mm}}

\begin{document}



\title{ {Arithmetic Average Density Fusion - Part IV:}
Distributed Heterogeneous Fusion of RFS and LRFS Filters via Variational Approximation
}

\author{Tiancheng~Li,~\IEEEmembership{Senior Member,~IEEE}, Haozhe~Liang, Guchong Li, Jes\'us Garc\'ia Herrero and Quan Pan %
\thanks{This work was partially supported 
 by National Natural Science Foundation of China (Grant No. 62071389 and 62201316), Natural Science Basic Research Program of Shaanxi (Program No. 2023JC-XJ-22) 
 and Fundamental Research Funds for the Central Universities. 
}
\thanks{T.\ Li, H. Liang, G. Li, and Q. Pan are all with the Key Laboratory of Information Fusion Technology (Ministry of Education), School of Automation, Northwestern Polytechnical University, Xi'an 710129, China, e-mail: \{t.c.li,guchong.li,quanpan\}@nwpu.edu.cn,haozheliang@mail.nwpu.edu.cn}
\thanks{J. G. Herrero is with the Group GIAA, University Carlos III of Madrid, Spain.
e-mail: jgherrer@inf.uc3m.es}
}

\maketitle

\begin{abstract}
This paper, {the fourth part of a series of papers on the arithmetic average (AA) density fusion approach and its application for target tracking,}
addresses the intricate challenge of {distributed} heterogeneous multisensor multitarget tracking, where each inter-connected sensor operates a probability hypothesis density (PHD) filter, a multiple Bernoulli (MB) filter or a labeled MB (LMB) filter 
and they cooperate with each other via information fusion. 
{Earlier papers in this series have}
proven that the proper AA fusion of these filters is all exactly built on averaging their respective unlabeled/labeled PHDs. Based on this finding, two PHD-AA fusion approaches are proposed via variational minimization of the upper bound of the Kullback-Leibler divergence between the local and multi-filter averaged PHDs subject to cardinality consensus based on the Gaussian mixture implementation, enabling heterogeneous filter cooperation. One focuses solely on fitting the weights of the local Gaussian components (L-GCs), while the other simultaneously fits all the parameters of the L-GCs at each sensor, both seeking average consensus on the unlabeled PHD, irrespective of the specific posterior form of the local filters. For the distributed peer-to-peer communication,
both the classic consensus and flooding paradigms have been investigated. 
Simulations have demonstrated the effectiveness and flexibility of the proposed approaches in both homogeneous and heterogeneous scenarios. 
\end{abstract}

\begin{IEEEkeywords}
Random finite set, arithmetic average fusion, distributed tracking, heterogeneous fusion, multitarget tracking
\end{IEEEkeywords}



\section{Introduction}
Since the pioneering work of the probability hypothesis density (PHD) filter \cite{Mahler03}, the random finite set (RFS) framework has garnered significant attention in the last two decades for multitarget tracking in clutter, which is particularly competitive in scenarios where the number of targets and their states are unknown and time-varying \cite{Mahler14book}.
Along with the maturity of advanced wireless sensing and communication technologies, the distributed multisensor RFS filtering has gained considerable attention \cite{Uney13,Battistelli13,Uney19consistency,SLi17LableInconsistence,Li19ParallelCC,Li20AAmb,Da_Li20TSIPN,Yi20AAfov,Gao20cphd,GLi22,GLi2He}, where netted sensors communicate and cooperate with each other in a peer-to-peer (P2P) fashion. 
In the classic setup for multi-sensor target tracking, the local filters are homogeneous, meaning that all local sensors run the same filtering algorithm; see the review \cite{Li22chapter}. 
However, the abundance of sensor networking technologies has led to the dominance of heterogeneous network \cite{Qiu18HeteroIoT,Yan20Heterogeneous,Yi21Heterogeneous} where the local sensors have unequal computing, memory capacities and isomeric measurements, and thus are suitable for running heterogeneous RFS filters, whether unlabeled or labeled RFS filters \cite{Vo15mtt}. 
The combination of heterogeneous filters, which can adopt diverse target and scenario models, inherently offers greater robustness and reliability compared to a uniform filter.
In fact, even within a homogeneous sensor network, employing different filtering/tracking algorithms can enhance robustness and reliability \cite{Yarvis05hetero, Liggins17,Li23Heterogeneous}. Nevertheless, the key challenge 
lies in the inability to directly fuse these heterogeneous multitarget densities while preserving their respective posterior form, not to mention coping with the complicated and unknown correlation between the inner-connected filters and be computing fast in the distributed case.


To be more specific, the heterogeneous RFS fusion we seek in this paper involves the classic PHD filter \cite{Mahler03,Vo05,Vo06}, the multiple Bernoulli (MB) filter \cite{Vo09CBmember}, the labeled MB (LMB) filter \cite{Vo13Label,Reuter14LMB}, 
all based on the popular Gaussian mixture (GM) implementation.
To address the heterogeneous fusion problem, our earlier consensus-driven fusion work \cite{Li23Heterogeneous} revised only the weights of the local Gaussian components (L-GCs) so that the corresponding (unlabeled) PHDs reach consensus approximately with each other, namely PHD-level arithmetic-average (AA) fusion \cite{Li22RFS-AA-Derivation}. That is, the consensus is reached on the first moment of the posteriors of these fusion filters, rather than on the posteriors themselves.
The fusion framework can be illustrated in Fig. \ref{fig:HFframework}, which provides a flexible, computationally efficient way for heterogeneous RFS filter cooperation. The statistics and properties of the AA density fusion including its insensitivity to arbitrary inter-sensor correlation and robustness to local misdetection have been studied in \cite{Bailey12,Li19Second,Li24SomeResults} in comparison with the geometric average density fusion approach {and more comprehensively} been analyzed in \cite{Koliander22}. Extension to fusion of {distributions over unknown quantities of interest} and of soft-decisions has been made in \cite{Kayaalp22} and \cite{Kayaalp23}, respectively.

\begin{figure*}
  \centering
  \includegraphics[width=17.5cm]{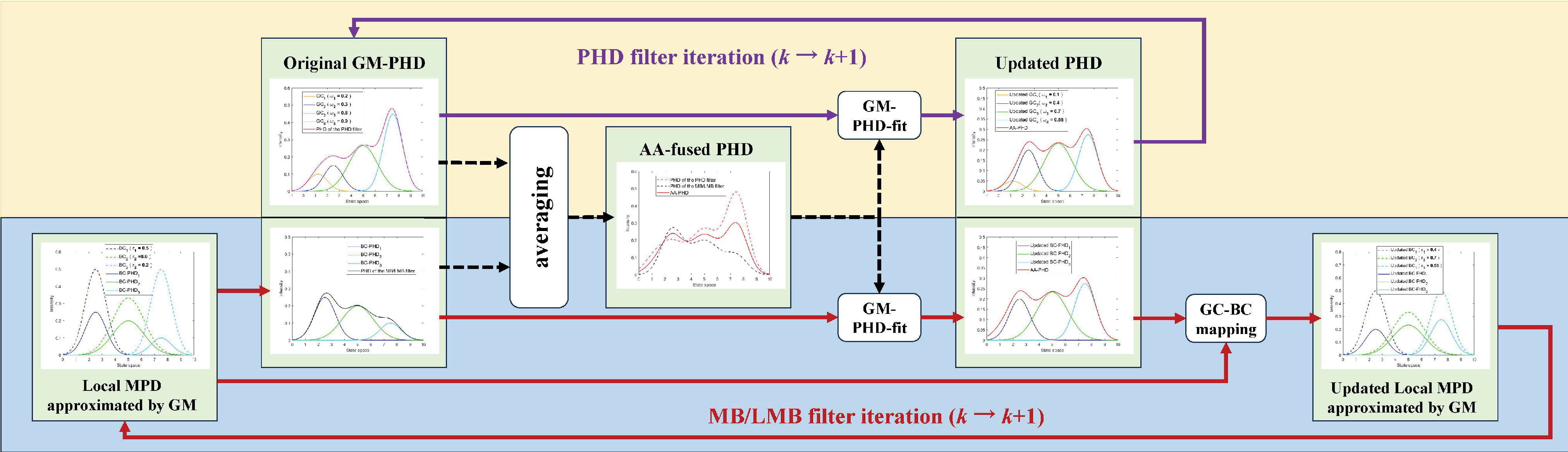}\\  
  \caption{Heterogeneous unlabeled and labeled RFS filter cooperation based on GM-PHD fit.} \label{fig:HFframework}
  \vspace{-2mm}
\end{figure*}

The heterogeneous GM-PHD-AA-fit solver \cite{Li23Heterogeneous}, however,
is a coordinate descent method (CDM) which is prone to non-convergence and computationally expensive, 
demonstrating significant limitations for large-scale networks.
To overcome these drawbacks, this paper employs the variational approximation (VA) approach to fit the weights and, additionally, the mean and covariance of the L-GCs for a more thorough GM-PHD-AA fusion, which results in two relevant fusion approaches of different algorithm complexities. In addition, we consider the distributed P2P sensor network, in which both the prevalent average consensus \cite{Olfati-Saber07,Sayed14book} and distributed flooding approaches \cite{Li17flooding} can be individually employed for P2P information communication. 



%


{This paper serves the fourth part of a series of papers that aim to provide a comprehensive and thorough study of the AA fusion methodology and its application for target tracking. The earlier three parts are available in \cite{Li24SomeResults,Li22RFS-AA-Derivation,Li23Heterogeneous}.}

The paper is organized as follows. 
The background is given in Section \ref{sec:background}. The proposed heterogeneous fusion approaches are detailed in Section \ref{sec:GMfit}. Alternative approximation methods are discussed in Section \ref{sec:bounds}.
Simulation is given in Section \ref{sec:simulation} before the paper is concluded in Section \ref{sec:conclusion}.

\section{Background and Preliminaries} \label{sec:background}
%

\subsection{Scenario Models and Assumptions}
The states of a random number of targets at time $k\in \mathbb{N}$ are described by an RFS $\Xi$, namely random target state set, with the form ${X}_k = \big\lbrace \mathbf{x}_{1}, \dots, \mathbf{x}_{n} \big\rbrace \in \mathbb{X}$, where $\mathbf{x}_{i} \in \mathbb{R}^x$ is the $i$-th target's state, $n =|{X}_k|$ denotes the random number of targets and $\mathbb{X}$ is the hyperspace of all finite subsets of the state space $\mathcal{X}$.
For any realization of $X_k$ with a given cardinality $n$, namely ${X}_k^n = \big\lbrace \mathbf{x}_{1}, \dots, \mathbf{x}_{n} \big\rbrace$, \cite[Eq.2.36]{Mahler14book},
the multitarget probability density (MPD) is defined as
$f({X}^n_k) = n! \rho(n) f(\mathbf{x}_{1}, \dots, \mathbf{x}_{n})$
where $f(\mathbf{x}_{1}, \dots, \mathbf{x}_{n} )$ are the localization densities and the cardinality distribution is given by $\rho(n)\triangleq \mathrm{Pr}\{|{{X}_k}|=n\} = \int_{|{X}_k| = n} {f({X}_k)\delta {X}_k}$. To distinguish different targets/tracks, the RFS is extended to the labeled RFS (LRFS) whose elements are assigned with distinct labels \cite{Vo13Label,Vo14GLMB}. 
Denote by $\mathbb{L}$ the hyperspace
of all finite subsets of the label space $\mathcal{L}$.
A realization of an LRFS $\widetilde{\Xi}$ with cardinality $n$, multitarget state ${X}^n_k$ and label set ${L}^n_k = \big\lbrace \ell_{1}, \ell_{2}, \dots, \ell_{n} \big\rbrace \in \mathbb{L}$ at time $k$ is denoted by $\widetilde{X}^n_k = \{ (\mathbf{x}_1, \ell_1),(\mathbf{x}_2, \ell_2),...,(\mathbf{x}_n, \ell_n) \} \in \mathbb{X} \times \mathbb{L}$.

The following scenario assumptions are made.

\subsubsection{Target motion}
The multitarget motion model is assumed as a union of several independent processes
$X_{k+1|k}=T_{k+1|k}(\mathbf{x}_1)\cup \dots \cup T_{k+1|k}(\mathbf{x}_n)\cup B_{k+1|k}$,
where the Bernoulli RFS $T_{k+1|k}(\mathbf{x}_i)$ is the set of targets at time $k+1$ generated by target $\mathbf{x}_i$ at time $k$, which is modeled by a survival probability $p^{\text{s}}_k(\mathbf{x}_i)$ and a Markov transition probability density function (PDF) $f_{k+1|k} (\cdot|\cdot)$, $B_{k+1|k}$ is the target birth RFS which might be modelled differently at local sensors based on either Poisson or MB RFS. 

\subsubsection{Sensor measurement} We consider an undirected P2P sensor network with $S$ nodes that are synchronous, perfectly coordinated and share the same region of interest (ROI).
For any $s \in \mathcal{S} =\{1,2,...,S\}$ with $\mathcal{S}$ denoting the sensor set, denote {$\mathcal{S}_s$} 
as the set of neighbor nodes to node $s$ including itself. 
The measurements of sensor $s$ at time $k$ are described by an RFS $\Sigma$ with the form ${Z}_{s,k} = \big\lbrace \mathbf{z}_{s,k,1}, \dots, \mathbf{z}_{s,k,m} \big\rbrace \in \mathbb{Z}$, where $\mathbf{z}_{s,k,m'} \in \mathbb{R}^z$ is the $m'$-th measurement generated by a target and collected by sensor $s$ at time $k$, $m$ and $\mathbb{Z}$ denote the number of measurements and  the hyperspace of all finite subsets of the measurement space $\mathcal{Z}$ respectively.
The measurement set can be decomposed by
$Z_{s,k}=\Upsilon_{1,k}(\mathbf{x}_1)\cup \dots \cup \Upsilon_{n,k}(\mathbf{x}_n)\cup C_{s,k}$, where the Bernoulli RFS $\Upsilon_{i,k}(\mathbf{x}_i)$ is the set of measurement generated by target $\mathbf{x}_i$ with detection probability $p_k^d(\mathbf{x}_i)$, 
and the clutter RFS $C_{s,k}$ follows a Poisson RFS and generates clutter measurement with rate $\lambda_{k}$. 



\subsection{First Moment of the MPD: PHD}

The RFS/LRFS is completely characterized by its MPD or labeled MPD; 
the detail of these unlabeled and labeled MPDs regarding the Poisson \cite{Mahler03,Vo05,Vo06}, MB \cite{Vo09CBmember} and LMB \cite{Vo13Label,Reuter14LMB} \cite{Vo13Label,Reuter14LMB} can be found in the literature and is omitted here. We hereafter focus on their (unlabeled) first-order moment density \cite[pp. 168-169]{Goodman97}, namely the PHD.
The PHD of the MPD $f_k(X_k)$ is an ordinary density
function on single target $\mathbf{x}\in \mathcal{X}$, defined as \cite[Ch.4.2.8]{Mahler14book} 
\begin{align}
	D_k(\mathbf{x})   
	& \triangleq \int_{\mathbb{X}}   {\bigg(\sum_{\mathbf{y}\in X_k}{\delta_\mathbf{y}}(\mathbf{x}) \bigg)f_k(X_k)\delta X_k}, \label{def-PHD}
\end{align}
where ${\delta_\mathbf{y}}(\mathbf{x})$ is the Dirac delta function concentrated at $\mathbf{y}$.

The (unlabeled) PHD of labeled MPD $f_k\big(\widetilde{X}_k\big)$ that is the joint distribution of the state and the label is given by \cite{Vo13Label}
\be
	D_k(\mathbf{x}) \triangleq \sum\limits_{l \in  \mathbb{L}} \int_{ \mathbb{X} \times \mathbb{L}}  {f_k}\big((\mathbf{x},l)\cup {\widetilde X_k}\big)\delta {\widetilde X_k},
\ee
where the integral of labeled RFS is defined in \cite{Vo13Label}.


\subsubsection{Poisson RFS} The PHD of the Poisson RFS $X_k$ with mean $\lambda_k$ and MPD $f_k^{\text {p}}(X_k) = {{\text{e}}^{ - {\lambda_k}}}\prod_{\mathbf{x} \in X} {{\lambda_k}{p}_k(\mathbf{x})}$ at time $k$ is given by
\begin{align}\label{eq:Poisson-PD}
	D_k^{\text{p}}(\mathbf{x}) =\,\, \lambda_k {p}_k(\mathbf{x}),
\end{align}
where $p_k(\mathbf{x})$ denotes the single-target probability density (SPD) at time $k$.

\subsubsection{Bernoulli RFS} The PHD of a Bernoulli RFS with target existence probability $r_k$ and SPD $p_k(\mathbf{x})$ is given by
\begin{align}
D_k^{\text{b}}\left(\mathbf{x}\right) & = r_kp_k\left(\mathbf{x}\right).
\label{eq:Bernoulli-PHD}
\end{align}

\subsubsection{MB RFS} An MB RFS $ X_k$ with $n$ independent Bernoulli RFSs \cite{Vo09CBmember} can represent maximum $n$ targets. 
Denoting the $l$-th Bernoulli component (BC) by $\big(r_k^{(l)},p_k^{(l)}(\mathbf{x})\big)$, the {PHD of MB RFS $ X_k$ with the MPD $f_k^\text{mb} (X_k) = \sum_{\uplus_{l=1}^{n}X^{(l)}=X}\prod_{l=1}^{n}f_k^\text{b} \left(X_k^{(l)}\right)$ is} given by
\begin{align}
	D_k^{\text {mb}}(\mathbf{x}) = \sum_{l=1}^n{r^{(l)}_k}{p^{(l)}_k}(\mathbf{x}),  \label{eq:PHD-MB}
\end{align}
where $\uplus$ denotes the disjoint union.

\subsubsection{LMB RFS}
The LMB can be viewed as a special MB with unique label assigned for each BC \cite{Reuter14LMB}. Given an MB RFS with parameter set $\big\{r_k^{(l)},p_k^{(l)}(\mathbf{x})\big\}_{l=1}^n$ at time $k$, we create a label set ${L}_k^n=\left\{\ell_1,...,\ell_n\right\}$ and $l$-th BC corresponds to a unique label $\ell_l\in {L}_k^n$. Hence, the LMB RFS can be characterized by parameter set $\big\{r_k^{(\ell_l)},p_k^{(\ell_l)}(\mathbf{x})\big\}_{\ell_l\in {L}_k^n}$.
The labeled PHD of track with label $\ell_l$ \cite{Vo13Label} is given by the weighted sum of densities for track with label $\ell_l$ over all hypotheses that contain the track with label $\ell_l$, namely
${\widetilde D}_k^{\text {lmb}}(\mathbf{x};\ell_l) =\mathbf{1}_{{L}_k^n}(\ell_l) \cdot r_k^{(\ell_l)} \cdot {\widetilde{p}}_k(\mathbf{x};\ell_l) \label{eq:LMB-Lphd}
$, where 
$\mathbf{1}_{L_k^n}(\ell_l)$ is the indicative function of the set $L_k^n$. 
The unlabeled PHD that is the sum of the
PHDs of all tracks \cite{Vo13Label} is given by
\begin{align}
	{D}_k^{\text {lmb}}(\mathbf{x}) &= \sum\limits_{\ell_l \in L^n_k } {\widetilde D}_k^{\text {lmb}}(\mathbf{x};\ell_l) = \sum\limits_{\ell_l \in L^n_k } r_k^{(l)} {p}_k(\mathbf{x};\ell_l). \label{eq:LMB-phd} 
\end{align}

\subsection{GM-PHD of Poisson, MB/LMB RFSs} \label{sec:Classic-RFS-distributions}
What follows presents the main result of the GM-PHD implementation of the PHD \cite{Vo06}, MB \cite{Vo09CBmember} or LMB \cite{Reuter14LMB} filters while their corresponding MPDs and PHDs are summarized in Tabel \ref{fig:RFS density}; detailed derivation can be found in the references.

\subsubsection{Poisson RFS}

The PHD filter is derived by assuming that the posterior follows a Poisson RFS but what is propagated over time is only the PHD rather than any MPD. 
The GM approximation of the PHD obtained by sensor $s \in \mathcal{S}$ at time $k$ that is completely determined by the parameter set $\mathcal{G}_{s,k} \triangleq \big\{ \big( \omega_{s,k}^{(j)} , \bm{\mu}_{s,k}^{(j)} , \bm{\Sigma}_{s,k}^{(j)}  \big) \big\}_{j=1}^{J_{s,k}}$, can be written as \cite{Vo06}:
\be\label{eq:GM_PHD}
	D_{s, k}(\mathbf{x}) \approx \sum_{j=1}^{J_{s,k}} \omega_{s,k}^{(j)} \mathcal{N}(\mathbf{x};\bm{\mu}_{s,k}^{(j)},\bm{\Sigma}_{s,k}^{(j)}), 
\ee
where 
$J_{s,k}$ is the number of GCs, $\omega_{s,k}^{(j)}$, $\bm{\mu}_{s,k}^{(j)}$ and $\bm{\Sigma}_{s,k}^{(j)}$ are the weight, mean and covariance of the $j$-th GC respectively. Correspondingly, the expected number of targets at sensor $s$ at time $k$ can be approximated by
\be
    \hat{N}_{s,k} \approx \sum_{j=1}^{J_{s,k}} \omega_{s,k}^{(j)}. \label{eq:N-est-sum_W}
\ee

\subsubsection{MB/LMB}
	
Denoting the  $\ell$-th BC of either the MB or the LMB at sensor $s$ by $\big(r_{s,k}^{(\ell)},p_{s,k}^{(\ell)}(\mathbf{x})\big), \ell \in {L}_{s,k}$, the PHD as given in \eqref{eq:PHD-MB} and \eqref{eq:LMB-phd} can be jointly rewritten as
	\begin{align}
		D_{s,k}^{\text {mb}}(\mathbf{x})  = \sum_{\ell \in {L}_{s,k}}{r_{s,k}^{(\ell)}}{p_{s,k}^{(\ell)}}(\mathbf{x}).  \label{eq:PHD-MB-ik}
	\end{align}

In the GM implementation, each BC $\ell$ is represented by $J_{s,k}^{(\ell)}$ GCs weighted by $\omega_{s,k}^{(\ell,\iota)} \!\ge\rmv 0$, $\iota=1,\dots,J_{s,k}^{(\ell)}$, i.e.,
$p_{s,k}^{(\ell)}(\mathbf{x}) = \sum _{\iota=1}^{J_{s,k}^{(\ell)}} \omega_{s,k}^{(\ell,\iota)} \ist \mathcal{N}\big(\mathbf{x};\bm{\mu}_{s,k}^{(\ell,\iota)}\rmv,\bm{\Sigma}_{s,k}^{(\ell,\iota)}\big) $ 
where
$\sum _{\iota=1}^{J_{s,k}^{(\ell)}} \omega_{s,k}^{(\ell,\iota)} = 1$.
This leads to the GM expression of the PHD \eqref{eq:PHD-MB-ik} as follows
\be\label{eq:mbPHD_GM}
	{D_{s,k}}(\mathbf{x}) \approx \sum_{\ell \in {L}_{s,k}} r_{s,k}^{(\ell)} \sum _{\iota=1}^{J_{s,k}^{(\ell)}} \omega_{s,k}^{(\ell,\iota)} \mathcal{N}\big(\mathbf{x};\bm{\mu}_{s,k}^{(\ell,\iota)}\rmv,\bm{\Sigma}_{s,k}^{(\ell,\iota)}\big),
\ee
which may be expressed as a unified GM determined by the parameter set $\mathcal{G}_{s,k} \triangleq \big\{ \mathcal{G}^{(\ell)}_{s,k} \big\}_{\ell \in {L}_{s,k}}$ as in \eqref{eq:GM_PHD}
with the GM size $J_{s,k} = \sum_{\ell \in {L}_{s,k}}  J_{s,k}^{(\ell)}$ and the weight of each GC reordered $j=1,\dots,J_{s,k}$ given as
\be\label{eq:mbPHD_WeightGC}
	\omega_{s,k}^{(j)} = r_{s,k}^{(\ell)}  \omega_{s,k}^{(\ell,\iota)},
\ee
which uses a unique GC index mapping 
\be\label{eq:indexMapping}
	j \leftrightarrow (\ell,\iota),
\ee
where $j=1,...,J_{s,k}, \ell  \in {L}_{s,k}, \iota=1,...,J^{(\ell)}_{s,k}$.


\begin{table*}[!htbp]
    \caption{Statistical characteristic of Poisson, MB and LMB}
    \centering

    \begin{tabular}{cccc}
    \toprule
         RFS&MPD&PHD&GM-PHD \\
    \midrule
         Poisson&$e^{-\lambda_k}\prod \limits_{\mathbf{x} \in X_k} \lambda_k {p}_k(\mathbf{x})$&$\lambda_k {p}_k(\mathbf{x})$&$\sum_{j=1}^{J_{s,k}} \omega_{s,k}^{(j)} \mathcal{N}(\mathbf{x};\bm{\mu}_{s,k}^{(j)},\bm{\Sigma}_{s,k}^{(j)})$\\
         MB&$\prod \limits_{\ell=1}^n (1-r_k^{(\ell)}) \sum \limits_{1 \le \ell_{1} \ne \dots \ne \ell_{n} \le n} \frac{{r^{(\ell_1)}_k}{p^{(\ell_1)}_k}(\mathbf{x})}{1-r_k^{(\ell_1)}} \dots \frac{{r^{(\ell_n)}_k}{p^{(\ell_n)}_k}(\mathbf{x})}{1-r_k^{(\ell_n)}}$&$\sum_{\ell=1}^n{r^{(\ell)}_k}{p^{(\ell)}_k}(\mathbf{x})$&$\sum_{\ell \in {L}_{s,k}} r_{s,k}^{(\ell)} \sum _{\iota=1}^{J_{s,k}^{(\ell)}} \omega_{s,k}^{(\ell,\iota)} \mathcal{N}\big(\mathbf{x};\bm{\mu}_{s,k}^{(\ell,\iota)}\rmv,\bm{\Sigma}_{s,k}^{(\ell,\iota)}\big)$\\
         LMB&$\delta_{|{\widetilde X}_k|}(|L^n_k|)\omega(L^n_k) \prod \limits_{(\mathbf{x},l) \in {\widetilde X}_k} {\widetilde p}_k(\mathbf{x},l)$&$\sum\limits_{l \in L^n_k } r_k^{(l)} {p}_k(\mathbf{x};l)$&$\sum_{\ell \in {L}_{s,k}} r_{s,k}^{(\ell)} \sum _{\iota=1}^{J_{s,k}^{(\ell)}} \omega_{s,k}^{(\ell,\iota)} \mathcal{N}\big(\mathbf{x};\bm{\mu}_{s,k}^{(\ell,\iota)}\rmv,\bm{\Sigma}_{s,k}^{(\ell,\iota)}\big)$\\
    \bottomrule

    \end{tabular}
    \label{fig:RFS density}

\end{table*}

\subsection{PHD-AA Fusion}
We consider the set of sensors ${\mathcal{S}'_s} \subseteq \mathcal{S}$ which share information with sensor $s$ by P2P
communication including sensor $s$ itself,
the AA fusion of their PHDs $\{D_r(\mathbf{x})\}_{r \in \mathcal{S}'_s}$ is given as follows
\be\label{eq:PHD-AA}
{D^{\text{AA}}_{\mathcal{S}'_s}}(\mathbf{x}) \triangleq \sum\limits_{r \in {\mathcal{S}'_s}} {{w_r}{D_r}(\mathbf{x})},
\ee
where the fusion weights $\sum\limits_{r \in {\mathcal{S}'_s}}{w_r}= 1, w_r > 0, \forall r \in \mathcal{S}$.

The PHD is a key statistical character of any RFS distribution, whose integral in any region gives the estimated number of target in that region \cite{Mahler14book}. Therefore, the PHD-AA fusion implies averaging the locally estimated numbers of targets, namely the cardinality consensus (CC) \cite{Li19CC}, i.e.,
$\hat{N}^{\text{AA}}_{\mathcal{S}'_s} = \int_{\mathcal{X}} {D^{\text{AA}}_{\mathcal{S}'_s}}(\mathbf{x})d \mathbf{x} = \int_{\mathcal{X}} \sum\limits_{r\in {\mathcal{S}'_s}} {w_r D_r(\mathbf{x})}  d \mathbf{x} =  \sum\limits_{r \in {\mathcal{S}'_s}} {w_r \hat{N}_r}$,
where $\hat{N}^{\text{AA}}_{\mathcal{S}'_s}$ denotes the average estimated number of targets and $\hat{N}_{r}$ denotes the estimated number of targets obtained from the PHD of the filter at sensor $r \in {\mathcal{S}'_s}$.

The AA fusion is a Fr\'{e}chet mean in the sense of the integrated
squared difference (ISD) \cite{Li20AAmb}, i.e.,
\be
    {D^{\text{AA}}_{\mathcal{S}'_s}}(\mathbf{x})  = \operatorname*{arg\,min}_{g\in\mathcal{F}_{\mathcal{X}}} \sum\limits_{r \in {\mathcal{S}'_s}}{w_r\int_{\mathcal{X}} \big(D_r(\mathbf{x})-g(\mathbf{x})\big)^2 d\mathbf{x}} \ist,
\ee
where $\mathcal{F}_\mathcal{X} \triangleq \{f: \mathcal{X} \rightarrow \mathbb{R} \}$.

It is also known as the best fit of the mixture (BFoM) subject to the CC constraint \cite{Li22RFS-AA-Derivation}; the proof is given in Appendix \ref{sec:app-PHD-BFoM} 
\begin{align}
 {D^{\text{AA}}_{\mathcal{S}'_s}}(\mathbf{x}) &= \operatorname*{arg\,min}_{g\in\mathcal{F}_\mathcal{X}} \sum\limits_{r \in {\mathcal{S}'_s}} {w_r\text{KL}\big(D_r\|g\big)},\label{eq:RFS-AA-Whole-KLD} \\
  \textit{s.t.} \hspace{3mm} & \int_{\mathcal{X}}g(\mathbf{x})d\mathbf{x}= \int_{\mathcal{X}}{D^{\text{AA}}_{\mathcal{S}'_s}}({\mathbf{x}})d\mathbf{x}, \label{eq:CC-constraint-BFoM}
\end{align}
where the Kullback-Leibler (KL) divergence, {extended to the PHD domain hereafter,} is defined as
\be
	\text{KL}\big(f\|g \big) \triangleq \int_\mathcal{X} {f(\mathbf{x})\log \frac{f(\mathbf{x})}{g(\mathbf{x})} d\mathbf{x}} \ist. \label{eq:KLD-def}
\ee

\section{Variational GM-PHD Fit} \label{sec:GMfit}
We here address averaging the GM-PHDs of heterogeneous RFS/LRFS filters, where the GM parameter set at time $k$ for sensor $s \in \mathcal{S}$ is denoted by $\mathcal{G}_{s,k}$ and the corresponding local PHD $D_{s,k}(\mathbf{x})$. 
The idea 
is to update the local GM parameters such that the corresponding PHD best fits the weighted PHD-AA $D^{\text{AA}}_{\mathcal{S}'_s,k}$ calculated by \eqref{eq:PHD-AA} (or approximately in the network average consensus manner), regardless the isomerism of the local filters.
To be more specific, {the approach} does not create or disregard any L-GCs in the local filters 
but adjust/optimize their weights {merely as what has been done in \cite{Li23Heterogeneous}} or jointly with the GC mean and covariance parameters, {both} via VA. The detail for obtaining $D^{\text{AA}}_{\mathcal{S}'_s,k}$ {through} P2P communication and consensus calculation 
in a distributed manner will be addressed in Section \ref{sec:dist_Implementation}, where $\mathcal{S}'_s$ {denotes the sensor set involved in the fusion which} depends on the number of iterations of P2P communication. 


Formally speaking, in the GM-PHD fit block of the proposed heterogeneous RFS/LRFS filter fusion framework as shown in Fig. \ref{fig:HFframework}, the individual GM parameters $\mathcal{G}_{s,k}=\{\omega_{s,k}^j,\bm{\mu}_{s,k}^j,\bm{\Sigma}_{s,k}^j\}_{j=1}^{J_{s,k}}$ (or merely the weights $\{\omega_{s,k}^j\}_{j=1}^{J_{s,k}}$) in each local filter are updated in \eqref{eq:GM-PHD-fit} while maintaining the CC as in \eqref{eq:CC-constraint}, the non-negativity of the weights as in \eqref{eq:non-negativity-w}, and the symmetry \eqref{eq:P_symmetry} and positive definiteness \eqref{eq:_non-negative} of the covariance,
\begin{align}
\mathcal{G}_{s,k}^{\text{BFoM}} & = \operatorname*{arg\,min}_{\mathcal{G}_{s,k}}  \mathbb{D}\big(D^{\text{AA}}_{\mathcal{S}'_s,k} \| D_{s,k} \big) \ist, \label{eq:GM-PHD-fit}\\
\textit{s.t.} \hspace{5mm} & \sum_{j=1}^{J_{s,k}} \omega_{s,k}^{(j)} =  \hat{N}^{\text{AA}}_{\mathcal{S}'_s,k} \ist, \forall s \label{eq:CC-constraint}\\
& \omega_{s,k}^{(j)} \geq 0, \forall s, j \ist, \label{eq:non-negativity-w}\\
&(\bm{\Sigma}_{s,k}^{(j)})^{\top}  =\bm{\Sigma}_{s,k}^{(j)} \ist, \forall s, j \label{eq:P_symmetry}\\
& \mathbf{x}^{\top} \bm{\Sigma}_{s,k}^{(j)} \mathbf{x} > 0, \forall s, j, \mathbf{x}\in \mathbb{R}^x \ist. \label{eq:_non-negative}
\end{align}
Here, $\mathbb{D}(f||g)$ is a discrepancy measure between $f$ and $g$ for which the focus of this paper is the KL divergence and its approximations/bounds. 

Our previous work \cite{Li23Heterogeneous} used the ISD metric as the optimization function \eqref{eq:GM-PHD-fit}, which allows for an analytical expression for GMs. It, however, is a non-convex function and the solver based on the CDM is susceptible to non-convergence or even divergence. In contrast, we redefine the metric in \eqref{eq:GM-PHD-fit} as the KL divergence in this work to better comply with the BFoM property of the AA fusion {as shown in \eqref{eq:RFS-AA-Whole-KLD}}. 
Since the GM-KL-divergence does not have analytical solution \cite{Hershey07KldGMs}, we resort to minimizing its variational upper bound (VUB).
%
The VA refers to approximating an arbitrary probability distribution under a specific
probabilistic graphical model \cite{wainwright2008graphical}. It
has been earlier used to reduce the size of the MB mixture \cite{williams2014efficient,xia2021poisson} and to reduce the data association in multiple extended target tracking \cite{granstrom2019poisson}, both demonstrating remarkable performance. Indeed, what follows shows that it yields much better approximation accuracy and computational  efficiency than the CDM approach.


%

\subsection{Minimizing VUB of GM-KL Divergence} \label{sec:ise}

Without loss of generality, we denote the desired PHD-AA $D^{\text{AA}}_{\mathcal{S}'_s,k}$ as $D_{\mathcal{S}'_s}$
and the original local $D_{s,k}$ as $D_s$ in short, i.e.,
\begin{align} \label{spnote}
	D_{\mathcal{S}'_s}(\mathbf{x}) &\triangleq \sum_{a} \pi_a \mathcal{N}_a(\mathbf{x}),  \\\label{spnote2}
    D_s(\mathbf{x}) &\triangleq \sum_{b} \omega_{s,b} \mathcal{N}_{s,b}(\mathbf{x}),
\end{align}
where $\pi_a$, $\omega_{s,b}$ denote the weight of the $a$-th and $b$-th GC in $D^{\text{AA}}_{\mathcal{S}'_s,k}$ and $D_{s,k}$, respectively.

We first introduce variational parameters $\phi_{s,b|a}\ge 0$ and $\varphi_{s,a|b}\ge 0$ satisfying the constraints $\sum_{b} \phi_{s,b|a}= \pi_a$ and $\sum_{a} \varphi_{s,a|b}= \omega_{s,b}$. For variational inference, \eqref{spnote} and \eqref{spnote2} can be, respectively, rewritten as
\begin{align}
	D_{\mathcal{S}'_s}(\mathbf{x}) &= \sum_{a,b} \phi_{s,b|a} \mathcal{N}_a(\mathbf{x}), \\
	D_s(\mathbf{x}) &= \sum_{a,b} \varphi_{s,a|b} \mathcal{N}_{s,b}(\mathbf{x}).
\end{align}

To satisfy the constraints \eqref{eq:CC-constraint} and \eqref{eq:non-negativity-w}, the variational parameters $\phi_{s,b|a}$ and $\varphi_{s,a|b}$ are subject to the following constraints:
\be
	\begin{cases}\label{eq:VA_constraints}
		\phi_{s,b|a}\ge 0,\  \forall (a,b), \\
		\varphi_{s,a|b}\ge 0,\  \forall (a,b), \\
		\sum_{b} \phi_{s,b|a}= \pi_a,\  \forall a,\\
		\sum_{a} \varphi_{s,a|b}= \omega_{s,b},\  \forall b,\\
		\sum_{a,b}\phi_{s,b|a}=\sum_{a,b}\varphi_{s,a|b},
	\end{cases}
\ee
where the last constraint is from \eqref{eq:CC-constraint}.

Using the Jensen's inequality, the upper bound of the objective function given in \eqref{eq:GM-PHD-fit} using the KL divergence and variational parameters and subject to the CC constraint \eqref{eq:VA_constraints} can be given as follows \cite[Sec.8]{Hershey07KldGMs}
\begin{align}\nonumber
	&\mathrm{KL}(D_{\mathcal{S}'_s}||D_s)\\
&= \int D_{\mathcal{S}'_s}(\mathbf{x}) \  \text{log}\  \frac{D_{\mathcal{S}'_s}(\mathbf{x})}{D_s(\mathbf{x})}\  d\mathbf{x}\nonumber\\
	&= - \int D_{\mathcal{S}'_s}(\mathbf{x}) \ \text{log}\  \bigg(\sum_{a,b} \frac{\varphi_{s,a|b}\mathcal{N}_{s,b}(\mathbf{x})}{\phi_{s,b|a}\mathcal{N}_a(\mathbf{x})} \frac{\phi_{s,b|a}\mathcal{N}_a(\mathbf{x})}{D_{\mathcal{S}'_s}(\mathbf{x})} \bigg)\ d\mathbf{x} \nonumber\\
	&\le -\sum_{a,b} \phi_{s,b|a} \int \mathcal{N}_a(\mathbf{x}) \ \text{log}\ \Big( \frac{\varphi_{s,a|b}\mathcal{N}_{s,b}(\mathbf{x})}{\phi_{s,b|a}\mathcal{N}_a(\mathbf{x})}\Big) \ d\mathbf{x} \nonumber\\
	&=
\mathrm{KL}(\phi||\varphi) + \sum_{a,b} \phi_{s,b|a}\  \mathrm{KL}(\mathcal{N}_a||\mathcal{N}_{s,b}), \label{eq_upbound}
\end{align}
where the KL divergence of $\mathcal{N}_{s,b}(\mathbf{x}):=\mathcal{N}(\mathbf{x};\bm{\mu}_{s,b},\bm{\Sigma}_{s,b})$ relative to $\mathcal{N}_a(\mathbf{x}):=\mathcal{N}(\mathbf{x};\bm{\mu}_a,\bm{\Sigma}_a)$ is derived from \eqref{eq:KLD-def} as
\begin{align} \label{eq:KLD-gaussian}
	\mathrm{KL}(\mathcal{N}_a||\mathcal{N}_{s,b})
	=&\frac{1}{2}\big[ (\bm{\mu}_a-\bm{\mu}_{s,b})^\text{T} \bm{\Sigma}_{s,b}^{-1} (\bm{\mu}_a-\bm{\mu}_{s,b}) \nonumber \\  & + \text{tr}(\bm{\Sigma}_{s,b}^{-1}\bm{\Sigma}_a)
	+\text{log} \frac{\text{det}(\bm{\Sigma}_{s,b})}{\text{det}(\bm{\Sigma}_a)}-n_x \big],
\end{align}
where $n_x$ is the dimension of $\mathbf{x}$.



Based on the VUB \eqref{eq_upbound}, the GM-PHD-AA fusion \eqref{eq:GM-PHD-fit} using the KL divergence is now given by determining the variational parameters $\{\phi_{s,b|a}, \varphi_{s,a|b}\}_{a,b}$ and L-GC parameters $\{\bm{\mu}_{s,b},\bm{\Sigma}_{s,b}\}_{b}$ that minimize the upper bound of the GM-KL-divergence \eqref{eq_upbound} at each filter $s\in \mathcal{S}$, i.e., 
\vspace{0.2cm}
\begin{problem}
\begin{align}
	\min \limits_{\{\phi_{s,b|a}, \varphi_{s,a|b}\},\{\bm{\mu}_{s,b},\bm{\Sigma}_{s,b}\}} &\mathrm{KL}(\phi||\varphi) + \sum_{a,b} \phi_{s,b|a}\  \mathrm{KL}(\mathcal{N}_a||\mathcal{N}_{s,b})
\end{align}
subject to \eqref{eq:VA_constraints}. 
\end{problem}

\begin{lemma}\label{varlemma}
Problem 1 can be solved by alternately updating the variational parameters $\left\{\phi_{s,b|a},\varphi_{s,a|b}\right\}_{a,b}$ 
and L-GC parameters $\{\bm{\mu}_{s,b},\bm{\Sigma}_{s,b}\}_b$ with guaranteed convergence. 
	\begin{proof}
		See Appendix \ref{sec:VAconvergence}.
	\end{proof}
\end{lemma}

%

In the next two subsections, we will address the optimization of the variational parameters $\{\phi_{s,b|a}, \varphi_{s,a|b}\}_{a,b}$ and $\{\bm{\mu}_{s,b},\bm{\Sigma}_{s,b}\}_b$, respectively.

\subsection{Variational GC-weight-fit}\label{sec:wfit}
We first keep the L-GCs $\{\mathcal{N}_{s,b}\}$ unchanged and only optimize their weights via $\{\phi_{s,b|a}, \varphi_{s,a|b}\}$ in Problem 1, which then simplifies to
\vspace{0.2cm}
\begin{problem} \label{problem2}
	\begin{align}\label{eq:VAPfit}
		\mathop{\min}\limits_{\{\phi_{s,b|a}, \varphi_{s,a|b}\}} \ &\mathrm{KL}(\phi||\varphi) + \sum_{a,b} \phi_{s,b|a}\  \mathrm{KL}(\mathcal{N}_a||\mathcal{N}_{s,b})
	\end{align}
subject to \eqref{eq:VA_constraints}.
\end{problem}

\begin{lemma}\label{lemma_VPmin}
	Problem \ref{problem2} can be reformulated by setting $\phi_{s,b|a}=\varphi_{s,a|b}=h_{s,a,b}$ which simplifies to
	\begin{align}
		\min \limits_{h_{s,a,b}} \ &\sum_{a,b}h_{s,a,b}\  \mathrm{KL}(\mathcal{N}_a||\mathcal{N}_{s,b}) \label{eq:h_aj_opt}
	\end{align}
    subject to \eqref{eq:VA_constraints} and can be solved by
    \begin{align}\label{eq:h_aj_solvtion}
    	h_{s,a,b}=\begin{cases}
    		\pi_{a}, &b=\operatorname*{arg\,min}_{b} \mathrm{KL}(\mathcal{N}_a||\mathcal{N}_{s,b}), \\
    		0,&\text{otherwise},
    	\end{cases}
    \end{align}
    which finally leads to the new weights of L-GCs as follows:
\begin{align} \label{eq:VPw}
	\omega_{s,b}= \sum_{a} h_{s,a,b}, \ b=1,...,J_{s}.
\end{align}
\begin{proof}
Since $\varphi_{s,a|b}$ appears only in the first, non-negative term of the objective function of \eqref{eq:VAPfit}, 
Problem \ref{problem2} can be simplified by setting $\phi_{s,b|a}=\varphi_{s,a|b}=h_{s,a,b}, \forall s,a,b$ so that the first term is minimized to $0$. 
Then, \eqref{eq:h_aj_opt} can be solved directly by letting $h_{s,a,b}=\pi_{a}$ if $\mathcal{N}_{s,b}$ has smaller KL divergence with $\mathcal{N}_a$ than other L-GCs and $h_{s,a,b} = 0$ otherwise, namely \eqref{eq:h_aj_solvtion}.
\end{proof}
\end{lemma}

\begin{remark}\label{remark_VPmin}
The optimal solver of Lemma \ref{lemma_VPmin} can be implemented by a clustering procedure applied on all the GCs of $D^{\text{AA}}_{\mathcal{S}'_s}$ at each local filter. That is, each L-GC $\mathcal{N}_{s,b}(\mathbf{x})$ obtained by the local filter plays a role of the cluster center and all the other L-GCs $\mathcal{N}_a(\mathbf{x})$ are grouped into each cluster according to their KL divergence to each L-GC. Meanwhile, those that have too large divergence with the L-GCs 
may be grouped with no L-GC/cluster. Finally, the weights of all GCs in each cluster $(s,b)$ will be summed up, resulting in the desired updated weight $\omega_{s,b}$ of the L-GC $(s,b)$. 
\end{remark}
\begin{algorithm}[t]
	\caption{variational GC-weight-fit}
	\begin{algorithmic}
		\setlength{\lineskip}{3pt}
		\setlength{\lineskiplimit}{3pt}
		\STATE \textbf{Input} $\left\{\bm{\mu}_{s,b},\bm{\Sigma}_{s,b}\right\}_{b=1}^{J_s}$, $\left\{\pi_{a},\bm{\mu}_{a},\bm{\Sigma}_{a}\right\}_{a=1}^{J_\text{AA}}$
        \STATE \textbf{Output} $\left\{\omega_{s,b}\right\}_{b=1}^{J_s}$
		\STATE \hspace{0.5cm}\textbf{initialize} $\left\{\omega_{s,b}\right\}_{b=1}^{J_{s}} =0$
		\STATE \hspace{0.5cm}\textbf{for} $a$ \textbf{=} 1 \textbf{to} $J_\text{AA}$ \textbf{do in parallel}
		\STATE \hspace{1.0cm}$n=\operatorname*{arg\,min}_{b} \mathrm{KL}(\mathcal{N}_a||\mathcal{N}_{s,b}), b=1,...,J_s$.
		\STATE \hspace{1.0cm}$\omega_{s,n} = \omega_{s,n}+\pi_{a}$
		\STATE \hspace{0.5cm}\textbf{end for}
	\end{algorithmic}
\end{algorithm}

The pseudo-code for this approach, referred to \textit{variational GC-weight-fit}, is summarized in Algorithm 1. At the end of the algorithm, the fused weights will be used to replace the weight of GCs in the PHD filter directly, or replaced that of the corresponding GC in the relevant BCs in the case of the MB/LMB filter according to the mapping \eqref{eq:indexMapping}. 


\subsection{Variational GM-PHD-fit}\label{sec:gmfit}
Obviously, a better fit can be expected if the parameters of L-GCs are also updated along with their weights. We therefore further optimize the means and variances of L-GCs $\left\{\mathcal{N}_{s,b}\right\}$ 
 in the local GM-PHD. 
That is, given $\{\phi_{s,b|a},\varphi_{s,a|b}\}$ 
yielded in the solver of Problem 2, Problem 1 simplifies to
\vspace{0.2cm}
\begin{problem} \label{problem3} 
	\begin{align}
		\min \limits_{\{\bm{\mu}_{s,b},\bm{\Sigma}_{s,b}\}} \ &  \sum_{a,b} h_{s,a,b}\ \mathrm{KL}(\mathcal{N}_a||\mathcal{N}_{s,b})  \label{eq:Lemm-gmfit} \\
  \text{s.t.} \hspace{3mm} &  (\bm{\Sigma}_{s,b})^{\top}  =\bm{\Sigma}_{s,b} \ist, \nonumber \\
  & \mathbf{x}^{\top} \bm{\Sigma}_{s,b} \mathbf{x} > 0, \hspace{3mm}  \forall b=1,..., J_s, \mathbf{x}\in \mathbb{R}^x , \nonumber
	\end{align}
where $h_{s,a,b}$ is given in \eqref{eq:h_aj_solvtion}. 
\end{problem}
\vspace{0.2cm}
\begin{lemma}\label{lemma:gm-g-fit}
	Problem \ref{problem3} can be solved by, $\forall b$,
\begin{align}
	&\bm{\mu}_{s,b}=\frac{1}{\omega_{s,b}} \sum_{a} h_{s,a,b} \bm{\mu}_a, \label{eq:GMmerging_m} \\
	&\bm{\Sigma}_{s,b}=\sum_{a} \frac{h_{s,a,b}}{\omega_{s,b}} \big[\bm{\Sigma}_a+(\bm{\mu}_a-\bm{\mu}_{s,b})^\text{T} (\bm{\mu}_a-\bm{\mu}_{s,b})\big].\label{eq:GMmerging_P}
\end{align}
\end{lemma}
\begin{proof}
According to \eqref{eq:RFS-AA-Whole-KLD}, it is straightforward to have
\begin{align}
		f_{s,b}(\mathbf{x}) & \triangleq \sum_{a} h_{s,a,b} \mathcal{N}_a(\mathbf{x}) \\
 & = \min \limits_{g\in\mathcal{F}_\mathcal{X}} \sum_{a} h_{s,a,b}\  \mathrm{KL}(\mathcal{N}_a||g). \label{eq:gb-kldmin}
	\end{align}
It is also known that \cite[Theorem 2]{Runnalls07merge}
\be\label{eq:best-fit-GM-Gaussian}
  (\bm{\mu}_{s,b}, \bm{\Sigma}_{s,b})= \mathop{\arg\min}\limits_{(\bm{\mu},\bm{\Sigma})}\text{KL}\big(f_{s,b}\| \mathcal{N}(\bm{\mu},\bm{\Sigma})\big) \ist,
\ee
where $\bm{\mu}_{s,b}, \bm{\Sigma}_{s,b}$ are defined in \eqref{eq:GMmerging_m} and \eqref{eq:GMmerging_P}, respectively.

By separating the constraint from the objection function, we now decompose the minimization problem as follows
\begin{align}
     & \sum_{a,b} h_{s,a,b}\ \mathrm{KL}(\mathcal{N}_a||\mathcal{N}_{s,b})  = \nonumber \\ \label{lem3}
     & \sum_{b} \Big( \sum_{a} h_{s,a,b}\  \mathrm{KL}(\mathcal{N}_a||f_{s,b}^{*}) + \mathrm{KL}(f_{s,b}^{*}||\mathcal{N}_{s,b})\Big) \ist.
\end{align}
Hence, the minimization of the left hand in (\ref{lem3}) is equivalent to the minimization of the two items in its right hand that are properly defined as follows: $f_{s,b}^{*} \in \mathcal{F}_\mathcal{X}$ can be any function while $\mathcal{N}_{s,b} \in \mathcal{F}_\mathcal{X}$ is limited to the Gaussian function. According to \eqref{eq:gb-kldmin} and \eqref{eq:best-fit-GM-Gaussian}, $f_{s,b}^{*}$ is given by $f_{s,b}$ and $\mathcal{N}_{s,b}$ is determined by  \eqref{eq:GMmerging_m} and \eqref{eq:GMmerging_P}. This indicats that \eqref{eq:Lemm-gmfit} is solved by finding the best fit arbitrary distribution $f_{s}^{*}$ first and then the Gaussian that best fits $f_{s}^{*}$.
\end{proof}


\begin{algorithm}[t]
	\caption{variational GM-PHD-fit}\label{alg:GM-Fit}
	\begin{algorithmic}
		\setlength{\lineskip}{3pt}
		\setlength{\lineskiplimit}{3pt}
		\STATE \textbf{Input} $\left\{\bm{\mu}_{s,b},\bm{\Sigma}_{s,b}\right\}_{b=1}^{J_s}$, $\left\{\pi_{a},\bm{\mu}_{a},\bm{\Sigma}_{a}\right\}_{a=1}^{J_\text{AA}}$, $\gamma_\text{g}$ 
        \STATE \textbf{Output} $\left\{\omega_{s,b},\bm{\mu}_{s,b},\bm{\Sigma}_{s,b}\right\}_{b=1}^{J_s}$
		\STATE \hspace{0.5cm}\textbf{initialize} $i=0, K_0=0$
		\STATE \hspace{0.5cm}\textbf{repeat}
		\STATE \hspace{1.0cm}\textbf{initialize} $i=i+1$, $\left\{C_{s,b}\right\}_{b=1}^{J_{s}}=\emptyset$, $K_i=0$
		\STATE \hspace{1.0cm}\textbf{for} $a$ \textbf{=} 1 \textbf{to} $J_\text{AA}$ \textbf{do in parallel}
		\STATE \hspace{1.5cm}$n=\operatorname*{arg\,min}_{b} \mathrm{KL}(\mathcal{N}_a||\mathcal{N}_{s,b}), b=1,...,J_s$
		\STATE \hspace{1.5cm}$h_{s,a,n}=\pi_a$
		\STATE \hspace{1.5cm}$C_{s,n}=C_{s,n} \cup \left\{h_{s,a,n},\bm{\mu}_{a},\bm{\Sigma}_{a}\right\}$
		\STATE \hspace{1.5cm}$K_i=K_{i}+\mathrm{KL}(\mathcal{N}_a||\mathcal{N}_{s,n})$
		\STATE \hspace{1.0cm}\textbf{end for}
		\STATE \hspace{1.0cm}\textbf{for} $b$ \textbf{=} 1 \textbf{to} $J_{i}$ \textbf{do}
		\STATE \hspace{1.5cm}update $\left\{\omega_{s,b},\bm{\mu}_{s,b},\bm{\Sigma}_{s,b}\right\}$ with $C_{s,b}$ using \eqref{eq:GMmerging_m} and
        \STATE \hspace{1.5cm}\eqref{eq:GMmerging_P}
		\STATE \hspace{1.0cm}\textbf{end for}
            \STATE \hspace{1.0cm}Calculate $\gamma_i$ using \eqref{eq:KL-rate}
		\STATE \hspace{0.5cm}\textbf{Until} $\gamma_i \leq \gamma_\text{g}$
	\end{algorithmic}
\end{algorithm}

It can be found that the solver given by Lemma 3, referred to as \textit{varitioanl GM-PHD fit},  consists of two steps: 1) grouping the GCs received from the other sensors to the nearest L-GC in the local sensor; 2) merging the GCs in each group to a single GC so that the weight, mean and covariance of the L-GCs are updated. 
These two correlated steps may be carried out in an alternating manner for multiple iterations to approach the optimal solver. 
Obviously, the iterative optimization of the variational parameters and the GM parameters necessitates the establishment of a stopping criterion. 
For this purpose, we define the goodness of the fitting at iteration $i =1, 2,..., I$ as $K_{i,s}  \triangleq \sum_{a,b} h_{i,s,a,b}\  \mathrm{KL}(\mathcal{N}_a||\mathcal{N}_{i,s,b})$ and set a threshold, denoted by $\gamma_\text{g}$, on the reducing rate $\gamma$ of this goodness as follows
\be\label{eq:KL-rate}
  \gamma_{i,s} \triangleq \frac{|K_{i,s}-K_{i-1,s}|}{K_{i-1,s}},
\ee
where $h_{i,s,a,b}$ and $\mathcal{N}_{i,s,b}$ are the obtained variational parameters and the Gaussian distribution at fitting iteration $i$.
The pseudo-code of the variational GM-PHD-fit approach is summarized in Algorithm 2. 

\subsection{Comparison and Cooperation with ISD-CDM}\label{sec:discussion}

Similar to the ISD-CDM \cite{Li23Heterogeneous} which seeks minimizing the ISD based on the CDM, both the proposed variational GC-weight-fit and GM-PHD-fit approaches do not create or remove any new GCs in the local filters. The local filters are able to preserve the specific form of MPD whether it is Poisson, MB or LMB. 
However, the ISD-CDM approach is computationally expensive due to iterative fitting and is prone to over-fitting/non-convergence owing to the use of a non-convex function gradient, 
as discussed in \cite[Sec. 4]{Li23Heterogeneous}. This issue has been avoided in our variational approaches where each parameter is optimized by an exact solver with guaranteed convergence as shown in Appendix \ref{sec:VAconvergence}. 

Moreover, our proposed heterogeneous fusion approaches can be operated flexibly to different levels. 
That is, different sensors may run either variational GC-weight-fit or variational GM-PHD-fit, and in the latter case, they can carry out the variational fit for different numbers of iterations according to the computing capacities of the sensors. 
Furthermore, both fit approaches can be carried out in parallel and even cooperate with the ISD-CDM \cite{Li23Heterogeneous}. This leads to a higher level of heterogeneous fusion where different sensors perform heterogeneous filters that cooperate with each other in a heterogeneous mode by using different fusion algorithms. 
See Fig. \ref{fig:sych} for an illustration of such a cooperation of variational GC-weight-fit, GM-PHD-fit and ISD-CDM and see Section \ref{sec:hyb_Simu} for the simulation demonstration. 
\begin{figure}[t]
	\centering
	\includegraphics[width=8.5cm]{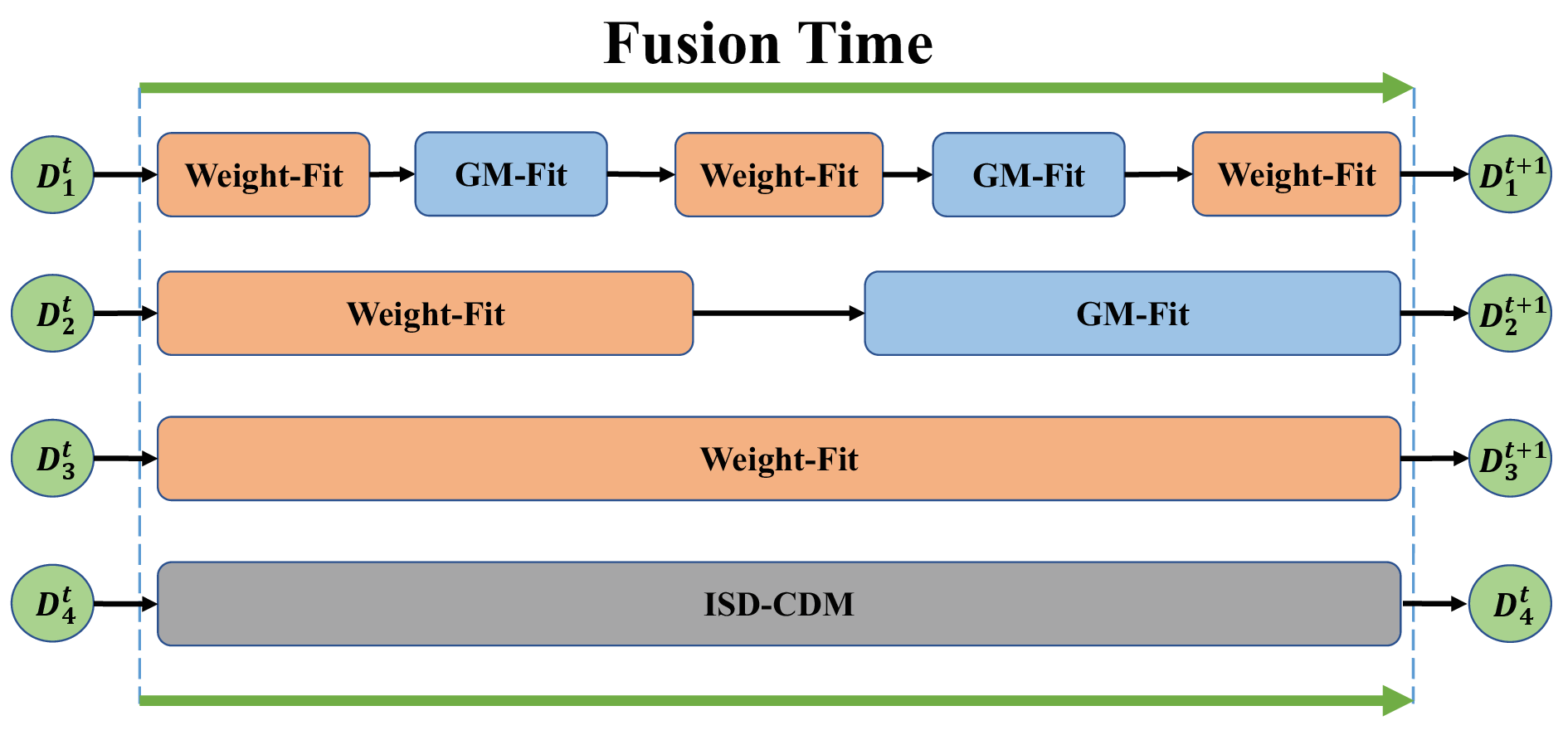}\\  
\vspace{-2mm}
	\caption{A high level of heterogeneous fusion where different nodes perform heterogeneous filters that cooperate with each other and may run different fusion algorithms in different numbers of fusion iterations. 
	} \label{fig:sych}
 \vspace{-2mm}
\end{figure}

\subsection{Distributed Implementation} \label{sec:dist_Implementation}
To accommodate the distributed sensor network, two P2P communication strategies including consensus and flooding can be employed. Both have been earlier detailed in the context of homogeneous PHD-AA fusion, e.g., \cite{Li17PC,Li17PCsmc,Li20AAmb,Li23AApmbm} and can be illustrated in Fig.~\ref{fig:consensus} and Fig.~\ref{fig:flooding}, respectively. We omit the details but highlight the following important issues:
\subsubsection{Fusion weight}
In the flooding algorithm, we assign uniform weights for aggregated PHDs from all involved sensors while in the consensus protocol, we employ the Metropolis weight \cite{xiao2006distributed} as the fusion weights, i.e., the PHD transformed from sensor $r \in {\mathcal{S}_{s}}$ to sensor $s$ will be assigned the following Metropolis weight in the fusion at sensor $s$
\be\label{metro}
	w_{r \to s}=\begin{cases}
		1/\max(|\mathcal{S}_{s}|,|\mathcal{S}_{r}|),\ &r\ne s \\
		1-\sum_{r\in \mathcal{S}_{s},r\ne s} w_{r \to s},\ &r=s
	       \end{cases}
\ee

The Metropolis weight was proposed for computing the unweighted average of all fusing items asymptotically and was proven of guaranteed convergence provided the infinitely occurring P2P communication. However, it simply does not guarantee closer result to the average in every communication and fusion step. Moreover, it suits homogeneous sensors that are of the similar quality. If the fusing items have significantly different qualities, the ultimate average should be accordingly weighted. We leave it as a part of our future work.

\subsubsection{VA errors}
In the flooding mode, the GM-PHD fit calculation is only carried out at the end of the P2P communication while it is carried out in each P2P communication iteration in the consensus protocol. In other words, the VA error is not brought in the flooding mode until at the end, but it occurs in each P2P iteration of consensus; see the difference illustrated in the comparison between Fig.~\ref{fig:consensus} and Fig.~\ref{fig:flooding}. 
Generally speaking,
the flooding algorithm that has a fast and deterministic convergence-to-consensus suffers from less VA errors in comparison with the consensus approach. 

\subsubsection{Communication cost}\label{sec:communicationcost}
We count the number of real values broadcast in the network at each iteration of fusion. Each GC is characterized by a weight, a $n_x$-dimensional mean vector and a $n_x \times n_x$ symmetric covariance matrix, taking communication of $1, n_x, \frac{n_x(n_x+1)}{2}$ real values, respectively. In the flooding communication mode, even only weights of the GCs are revised in the variational GC-weight-fit approach, all GM parameters need to be inter-node communicated as in the variational GM-PHD-fit approach; namely both fusion approaches are the same costly in communication. At each iteration of P2P communication, only the parameters of the GCs that have never been transformed between two neighbor sensors are communicated. In the consensus mode, the GC parameters (mean and covariance) remain unchanged in the variational GC-weight-fit approach and so they do not need to be repeatedly communicated between any two nodes; i.e., they are communicated in a flooding manner. But, the weights of the GCs will be updated in each iteration of the P2P fit calculation and need to be repeatedly inter-communicated as what has been done in the variational GM-PHD-fit approach where all parameters are communicated between neighbors in each P2P iteration.

\begin{figure}[t]
	\centering
	\includegraphics[width=7.8cm]{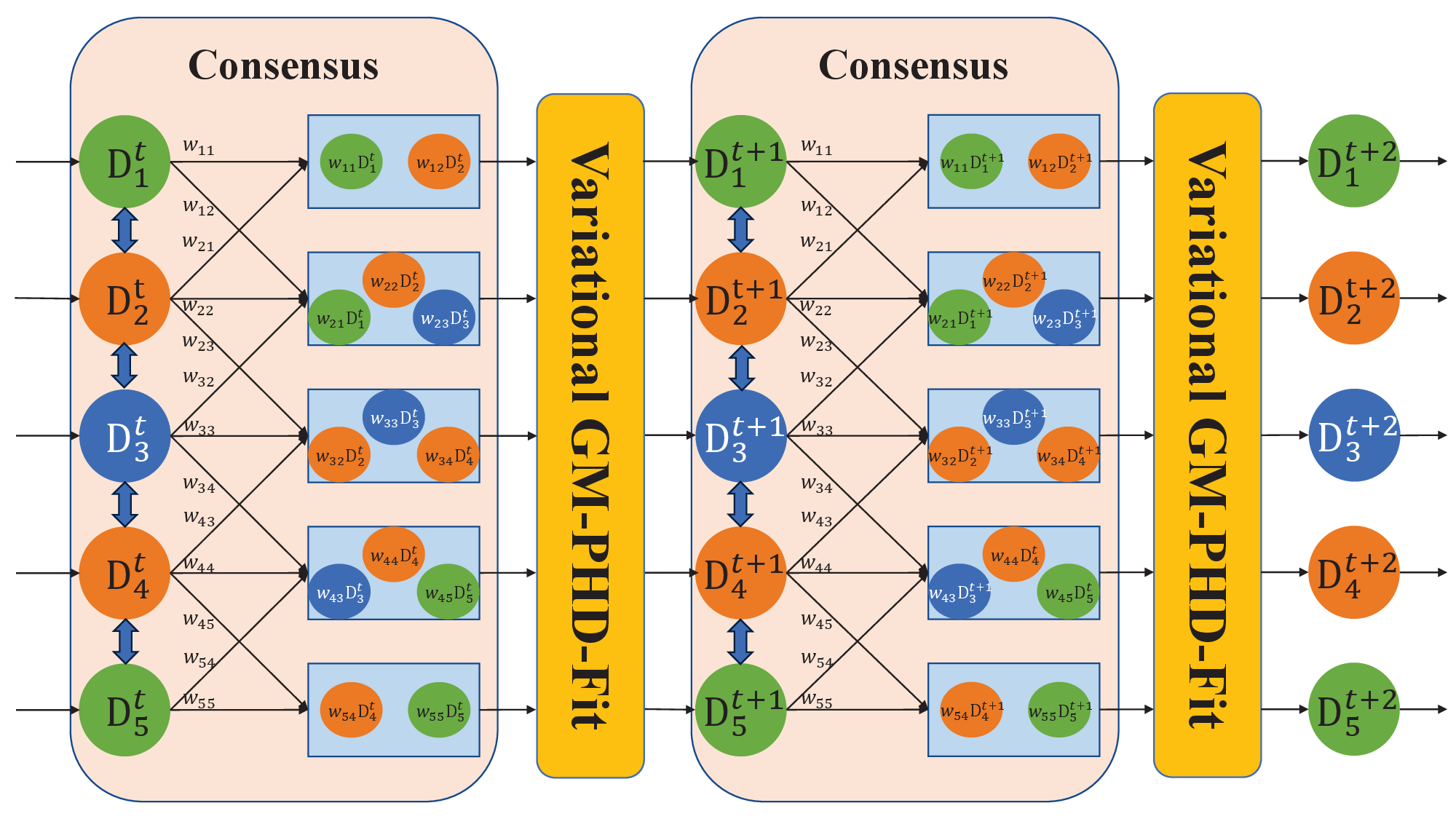}\\  
 \vspace{-1mm}
	\caption{Illustration of Consensus-based variational GM-PHD-fit. 
	} \label{fig:consensus}
	\vspace{-2mm}
\end{figure}


\begin{figure}[t]
	\centering
	\includegraphics[width=7.8cm]{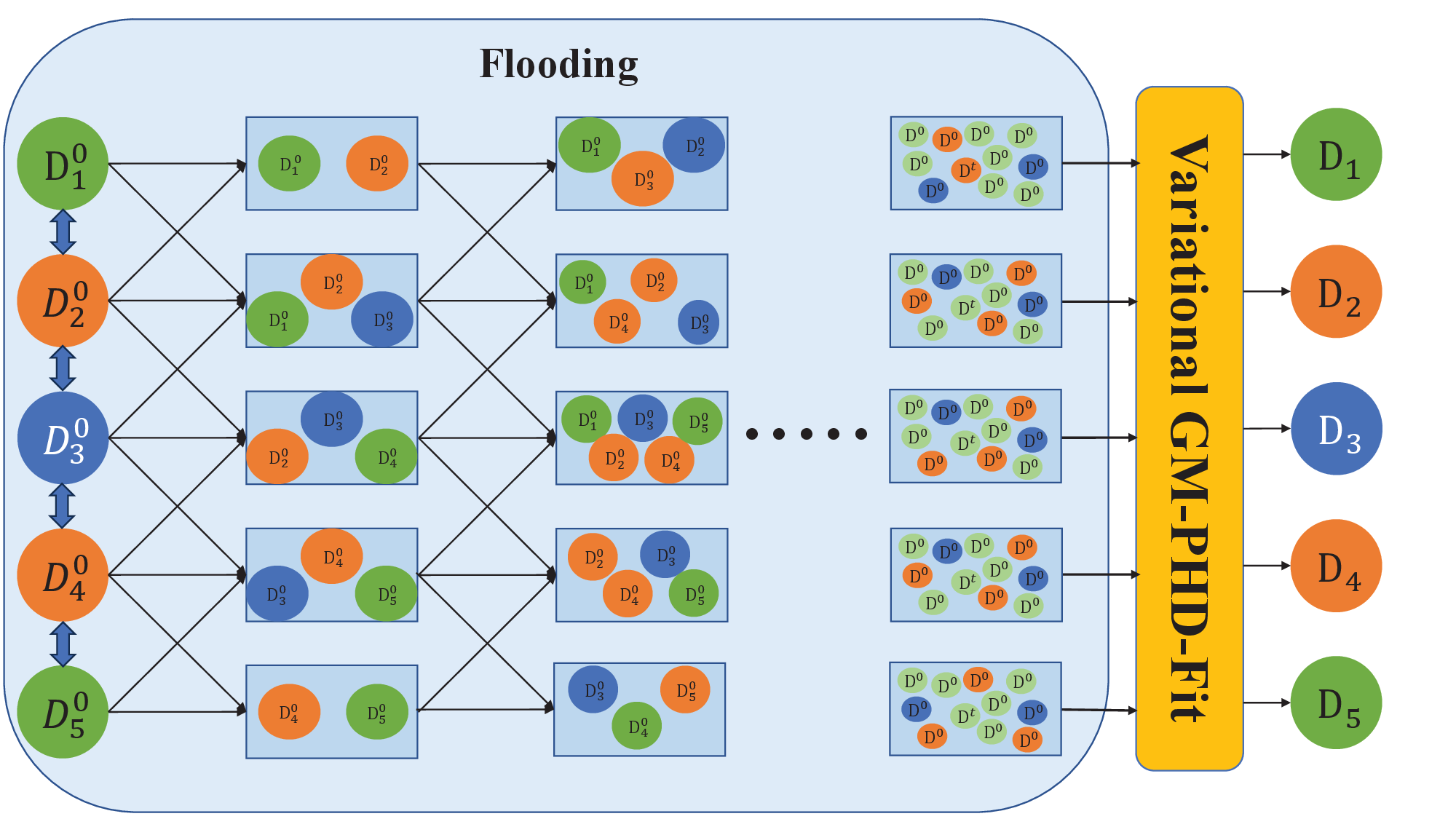}\\  
 \vspace{-1mm}
	\caption{Illustration of Flooding-based variational GM-PHD-fit. 
	} \label{fig:flooding}
	\vspace{-2mm}
\end{figure}


\section{Other Approximations} \label{sec:bounds}
There are some other approximations to the KL divergence between two GMs which may be employed as the objective function to be minimized in the proposed GM-PHD-fit framework. 
In this section, we consider five alternatives, which overall fall into two categories as follows:
\begin{itemize}
    \item The first two reduce the GM to a single GC which suits the case of single target only.
    \item The other three approximations seek different bounds on the KL divergence between GMs.
\end{itemize}
{Two} of them will be verified in comparison with the proposed VUB via simulation in Section \ref{sec:Bounds_Simu}.

\subsection{Gaussian Approximation}
A Gaussian approximation of the KL divergence between two GMs is given by \cite[Sec.4]{do2003fast}
\begin{align}
\mathbb{D}_1(D_{\mathcal{S}'_s}||D_s) \triangleq \mathrm{KL}(\hat{d}||\hat{g}),
\end{align}
where $\hat{d}$, $\hat{g}$ are the best fit Gaussian distribution with relative to $D_{\mathcal{S}'_s}$ and $D_s$, respectively; they are given by \eqref{eq:best-fit-GM-Gaussian} and the divergence is calculated as \eqref{eq:KLD-gaussian}. Another Gaussian approximation is given by using the nearest pair of Gaussian, i.e.,
\begin{align}
\mathbb{D}_2(D_{\mathcal{S}'_s}||D_s) \triangleq \min \limits_{a,b} \mathrm{KL}(\mathcal{N}_a||\mathcal{N}_{s,b}).
\end{align}
Both of them are improper for the GM-PHD-fit since GM-PHD represents the first moment of the multitarget density which can not be directly approximated by a single Gaussian.

\subsection{Convex Upper Bound}
A convex upper bound of the KL divergence between two GMs is given by \cite{goldberger2003efficient}
\begin{align}
\mathbb{D}_3(D_{\mathcal{S}'_s}||D_s) \le \sum_{a,b} \pi_a \omega_{s,b} \mathrm{KL}(\mathcal{N}_a||\mathcal{N}_{s,b}), \label{eq:Conv-upper-bound}
\end{align}
which is a naive upper bound of the KL divergence that tends to overestimate. Used in the objective of our proposed GM-PHD-fit, it tends to maximize the weight of $\mathcal{N}_{s,b}$ that has the smallest $\pi_a \omega_{s,b} \mathrm{KL}(\mathcal{N}_a||\mathcal{N}_{s,b})$ while other weights close to $0$, which is obviously unreasonable in our case.


\subsection{Variational Lower Bound}
We consider two variational lower bounds of the KL divergence of GMs. The first one is given by \cite[Sec.5]{do2003fast}
\begin{align}
\mathbb{D}_4(D_{\mathcal{S}'_s}||D_s) \ge \sum_{a} \pi_a \text{log} \frac{\sum_{a'}\pi_{a'} z_{aa'}}{\sum_{b}\omega_{s,b} z_{ab}},
\end{align}
where $z_{ab} = \int \mathcal{N}_a (\mathbf{x}) \mathcal{N}_{s,b} (\mathbf{x}) d\mathbf{x}$ is the scalar factor of the product of two Gaussians.

The above lower bound tends to underestimate. 
The second, a more accurate one with closed-form expression, is given by \cite[Sec. 7]{do2003fast}
\begin{align}
\mathbb{D}_5(D_{\mathcal{S}'_s}||D_s) \ge \sum_{a} \pi_a \text{log} \frac{\sum_{a'}\pi_{a'} e^{-\mathrm{KL}(\mathcal{N}_a||d_{a'})}}{\sum_{b}\omega_{s,b} e^{-\mathrm{KL}(\mathcal{N}_a||\mathcal{N}_{s,b})}}, \label{eq:Variational-lower-bound}
\end{align}
which does not satisfy the positivity property in general.

\section{Simulations} \label{sec:simulation}
The simulation ROI is given by $[-1\ist\text{km}, 1\ist\text{km}]\times [-1\ist\text{km}, 1\ist\text{km}]$ 
monitored by 12 sensors as depicted in Fig. \ref{fig:scenario}. The target state is denoted as $\mathbf{x}_k=[x_k \; \dot{x}_k \; y_k \; \dot{y}_k]^\text{T}\rmv$ with planar position $[x_k \; y_k]^\text{T}$ and velocity $[\dot{x}_k \; \dot{y}_k]^\text{T}$. The target birth is modeled by a GM form with four GCs (corresponding to the Poisson process for the PHD filter and the MB process for the MB and LMB filters) as follows:
\begin{align}\nonumber
    \gamma(\mathbf{x}) = \sum\limits_{\ell=1}^{4}r_{\text{B}}\mathcal{N}(\mathbf{x};\bm{\mu}_\text{B}^{(\ell)},\bm{\Sigma}_\text{B}),
\end{align}
where $r_{\text{B}} =0.03$, $\bm{\Sigma}_\text{B}^{\text{ \ \ \ }}=\mathrm{diag}([10\text{m},10\text{m/s},10\text{m},10\text{m/s}]^\text{T})^{2}$, $ \bm{\mu}_\text{B}^{(1)}=[0,0,0,0]^\text{T}$, $\bm{\mu}_\text{B}^{(2)}=[400\text{m},0,-600\text{m},0]^\text{T}$, $\bm{\mu}_\text{B}^{(3)}=[-800\text{m},0,-200\text{m},0]^\text{T}$, and $\bm{\mu}_\text{B}^{(4)}=[-200\text{m},0,800\text{m},0]^\text{T}$.

Each target has a constant survival probability $0.95$ and follows a constant velocity motion. For generating the ground truth, a noiseless transition density $f_{k|k-1}(\mathbf{x}_{k}|\mathbf{x}_{k-1})=\mathcal{N}(\mathbf{x}_{k};\mathbf{F} \mathbf{x}_{k},\mathbf{0}_{4\times4})$ is adopted, while the filters utilize $f_{k|k-1}(\mathbf{x}_{k}|\mathbf{x}_{k-1})=\mathcal{N}(\mathbf{x}_{k};\mathbf{F} \mathbf{x}_{k},\mathbf{Q})$
where
$$
\mathbf{F}=\mathbf{I}_{2}\otimes\left[\begin{array}{cc}
1 & \Delta\\
0 & 1
\end{array}\right], \,\mathbf{Q}=25 \times\mathbf{I}_{2}\otimes\left[\begin{array}{cc}
\Delta^{2}/2 & \Delta/2\\
\Delta/2 & \Delta
\end{array}\right], 
$$
with $\Delta$ and $\otimes$ denoting the sampling interval (one second in our case) and the Kronecker product operator, respectively.


The simulation is performed for 100 runs in total using the same ground truth with conditionally independent measurement series for 100 seconds each run.
For simplicity, all 12 sensors have the same time-invariant target detection probability $0.9$ and linear measurement made on the $x-y$ positions of the target state with
mutually independent zero-mean Gaussian noises with the same standard deviation of $10$m.
The clutter is uniformly distributed over the ROI and modeled as a Poisson process with rate $10$. 

For distributed implementations, various numbers of P2P iterations, denoted by $t$, over the sensor network have been considered. The simulation is carried out in either homogeneous or heterogeneous cases. In the former, all the sensors run the same PHD, MB or LMB filters while in the latter, different sensors may run different filters. In either case, the local filters may not cooperate with each other at all (namely noncooperative, indicated by $t=0$), only communicate and fuse with each other the estimated number of targets (namely CC only), or perform the proposed GM-PHD-fit via the P2P consensus or flooding scheme for various numbers of iterations $t=1,..., 4$. 
Both the ISD-CDM and the variational GM-PHD-fit approaches need to run a number of iteration steps for optimization in each fusion step.
The former uses learning rate $\alpha_i =0.2$ and fading rate $\beta_i=0.6$ for maximally $i_\text{max} =4$ fitting iterations \cite{Li23Heterogeneous} while the variational GM-PHD-fit approach use threshold $\gamma_\text{g}=0.1$ in Algorithm \ref{alg:GM-Fit}.  
%
The maximum number of L-GCs in the local GM is $200$ for the PHD filters, the maximum number of tracks/BCs is $50$ and the maximum number of L-GCs for each track/BC is $20$ in the MB/LMB filters. 
The other setup of the filters is the same as given in the codes released by Vo-Vo at https://ba-tuong.vo-au.com/codes.html.

\begin{figure}[t]
  \centering
  \includegraphics[width=8cm]{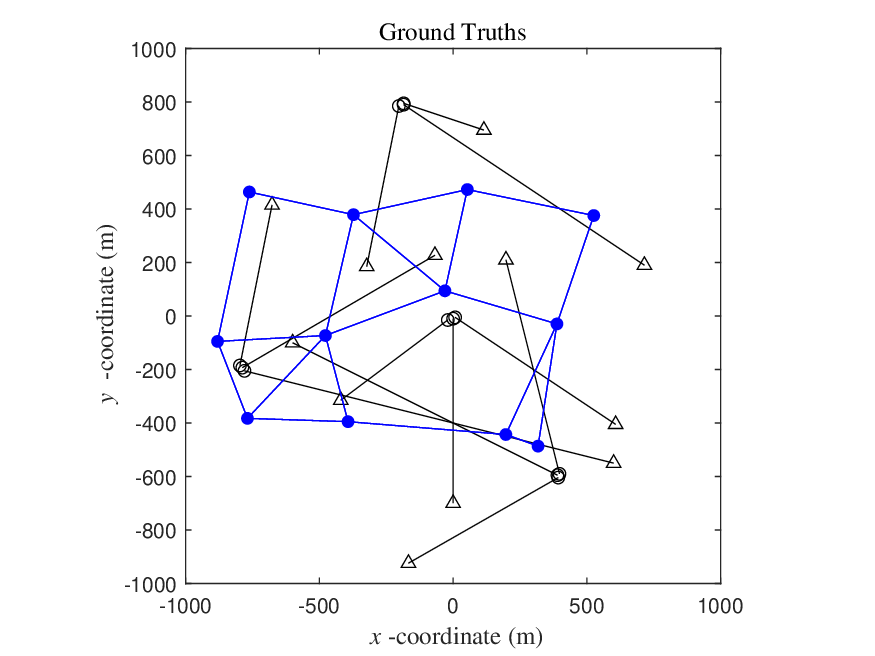}\\
  \vspace{-2mm}
  \caption{Sensor network (the blue points represent sensors and edges communication links) and target trajectories starting from $`\circ$' and ending at $`\triangle$'.} \label{fig:scenario}
  \vspace{-3mm}
\end{figure}

The filter accuracy is evaluated by the optimal subpattern assignment (OSPA) error \cite{Schuhmacher08}, 
which is given as follows, for $|{Y}| \geq |{X}|$,
\be 
\left( \frac{1}{|{Y}|}
\bigg( {\mathop {\min }\limits_{\pi  \in {{\rm \Pi} _{|{Y}|}}} \sum\limits_{i = 1}^{|{X}|} {{d^{(c)}}{{({\mathbf{x}_i},{\mathbf{y}_{\pi (i)}})}^p}} }  +  { {{c^p}}  (|{Y}| - |{X}|)} \bigg) \right)^{\frac{1}{p}} \nonumber
\ee
where $\pi$ and $ {\rm \Pi}_n $ are a permutation and the set of all permutations on $\{1,\ldots,n \}$, and $ {d^{(c)}}(\mathbf{x},\mathbf{y}) = \min \left( {d(\mathbf{x},\mathbf{y}),c} \right) $ is the Euclidean distance between $\mathbf{x}$ and $\mathbf{y}$ with threshold $c$. 
In our simulation, we use $c=100$m and $p=2$.

\subsection{Comparison of Different GM-KL Divergence Bounds} \label{sec:Bounds_Simu}
Firstly, we compare the convex upper bound \eqref{eq:Conv-upper-bound} and variational lower bound \eqref{eq:Variational-lower-bound}  with the VUB \eqref{eq_upbound} for GC-weight-fit, namely revising the weights of the L-GCs only.
Since we did not find a nice solver to the former two bounds as Lemma \ref{lemma_VPmin} for minimizing VUB, we employed the classic CDM method for approximating them.
All the sensors in the network run the PHD filter and communicate with each other in the flooding manner. Fig. \ref{fig:Bound} shows the average OSPA errors with different numbers of flooding iterations. Evidently, the proposed variational GC-weight-fit using VUB performs much better than the other bounds. In the following, we will focus merely on the VUB.

\begin{figure}
  \centering
  \includegraphics[width=7cm]{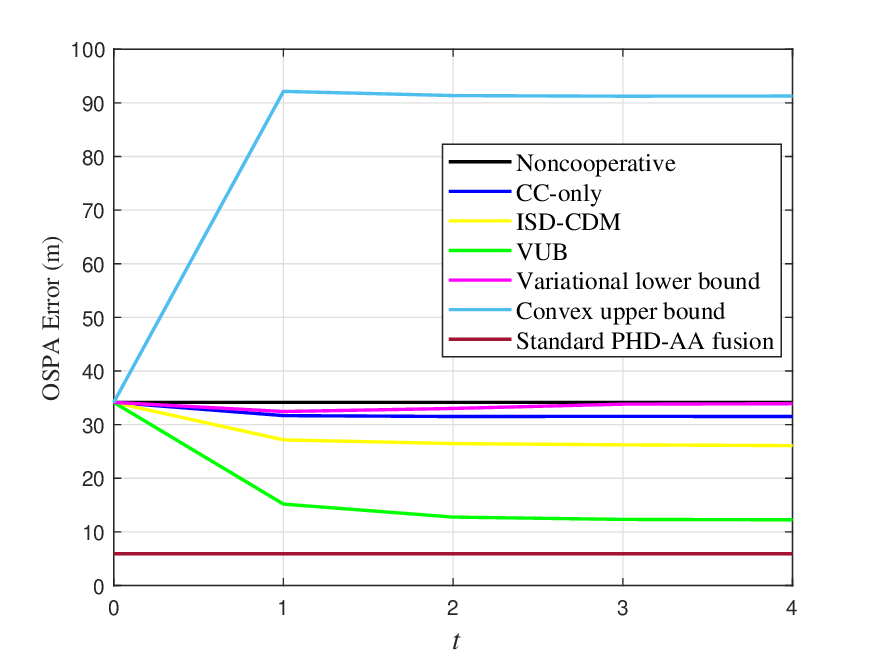}\\
  \vspace{-2mm}
  \caption{The performance of the GC-weight-fit in the case of homogeneous fusion of PHD filters via approaching the convex upper bound \eqref{eq:Conv-upper-bound}, variational lower bound \eqref{eq:Variational-lower-bound}  and the VUB \eqref{eq_upbound} for GC-weight-fit, respectively, in comparison with the ISD-CDM and CC only approaches.} \label{fig:Bound}
  \vspace{-3mm}
\end{figure}

\begin{figure*}[!t]
  \centering
  \includegraphics[width=15cm]{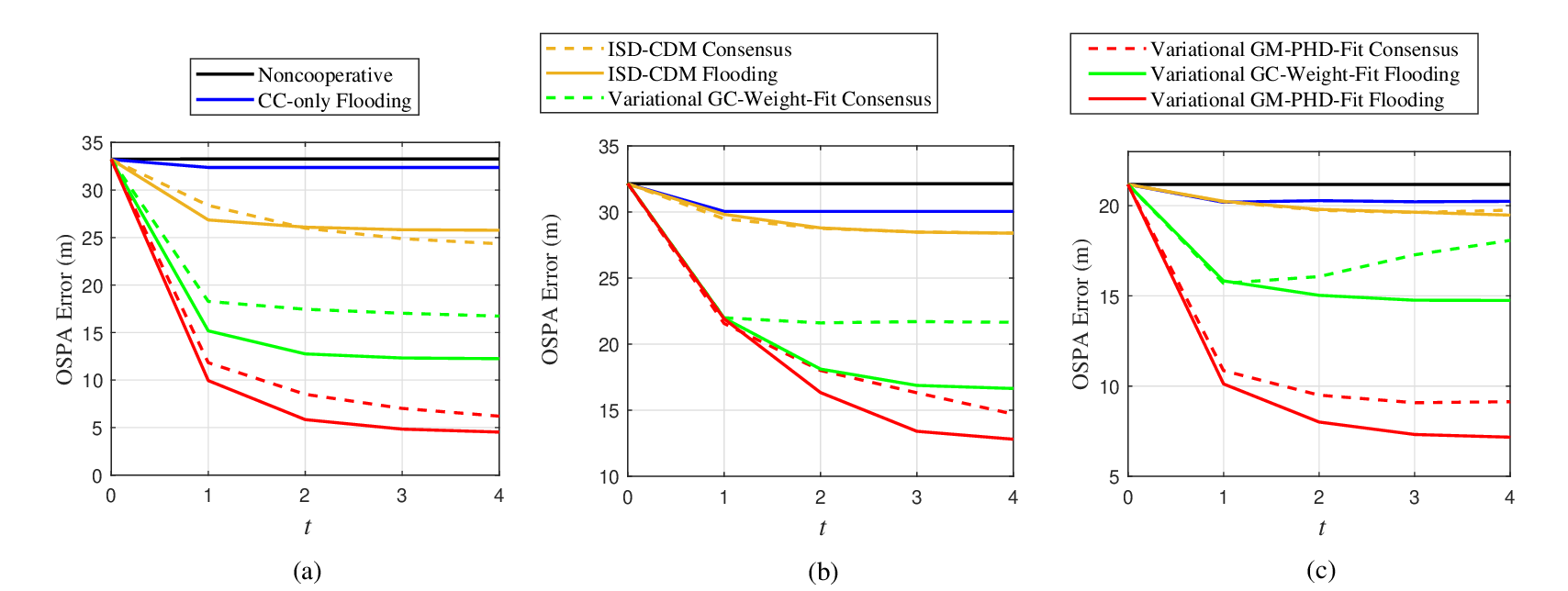}\\
  \vspace{-2mm}
  \caption{Average OSPA error of the homogeneous fusion of PHD (a), MB (b) and LMB (c) filters, respectively.} \label{fig:homo-Aver-opsa}
  \vspace{-2mm}
\end{figure*}

\begin{figure}
  \centering
  \includegraphics[width=6.5cm]{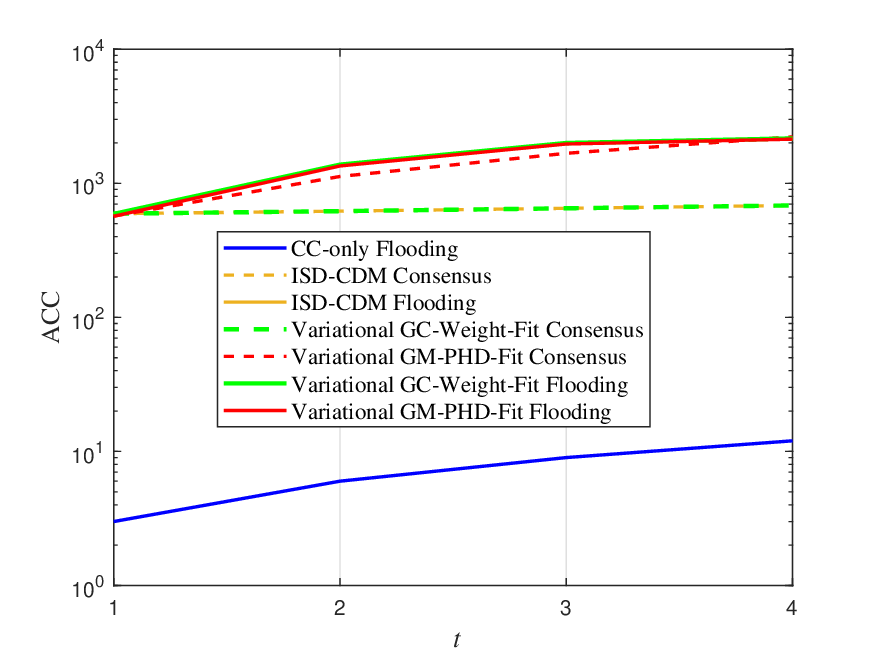}\\
  \vspace{-2mm}
  \caption{ACC of the PHD filters using different fusion approaches.} \label{fig:PHDcommunicationcost}
  \vspace{-2mm}
\end{figure}

\subsection{Homogeneous PHD/MB/LMB Fusion} \label{sec:homogeneousSim}
In this subsection, we test the performance of the proposed variational GM-PHD-fit approach in the homogeneous case of the PHD filter fusion. 
As shown in Fig. \ref{fig:homo-Aver-opsa} (a), the average OSPA errors of the local PHD filters over all 100 runs and all 100 filtering steps have been significantly reduced in both GC-weight-fit and GM-PHD-fit approaches, especially in the latter, using whether flooding or consensus. 
The average-over-time OSPA errors of the homogeneous fusion approaches for MB and LMB filters are given in Fig. \ref{fig:homo-Aver-opsa} (b) and (c), respectively. All fusion methods demonstrate a gradually decreased OSPA error with the increase of the number of P2P communication iterations, namely network consensus convergence, except for the GC-weight-fit consensus approach in the case of LMB filters. We conjugate that this is mainly because of the the approximate error of the consensus algorithm using the Metropolis weights, which together with the non-negligible VA error for LMB filter fusion, accumulates with the increase of consensus iterations. As a result, the OSPA error is significantly reduced at the first fusion iteration but slightly increased with the increase of the number of fusion iterations when $t\geq 2$. This, however, is mitigated in the case of variational GM-PHD-fit of which the fusion gain is more significant than the consensus divergence error. In summary,
\begin{itemize}
    \item The proposed variational GM-PHD-fit
 performs the best, the variational GC-weight-fit
 the second, the ISD-CDM the third and the CC-only the forth.
 \item The flooding approach converges faster, achieves lower OSPA error in comparison with than the consensus approach for each $t =1,...,4$.
\end{itemize}
More specifically, the variational fusion approaches perform much better than ISD-CDM methods and can reduce the OSPA error as much as 85\% in the case of variational GM-PHD-fit flooding. The reduction of the OSPA error can be as significantly as 61\% and 65\% in the MB and LMB cases, respectively. In contrast, the reduction in the cases of CC only, ISD-CDM is no more than 10\% and 25\%, respectively. As found within the ISD-CDM approach \cite{Li23Heterogeneous}, better fusion of the MBs/LMBs should be given in a proper Bernoulli-to-Bernoulli manner namely labeled PHD consensus.

Fig. \ref{fig:PHDcommunicationcost} presents the average communication cost (ACC) of local PHD filters, namely the number of real values broadcast by a sensor to its neighbors during all the dissemination/fusion iterations performed at one time step, averaged over all the sensors, time steps, and simulation runs. As shown, all flooding algorithms have the similar ACC except for the CC-only approach that broadcasts merely the estimated number of targets among sensors. Their slight difference is due to the different local GM sizes obtained after the PHD-AA fusion.
However, the ISD-CDM consensus and the proposed variational GC-weight-fit consensus have lower ACC than the others when $t\geq 2$ because only weights of GCs are changed and need to be broadcast.

\begin{figure}
  \centering
\includegraphics[width=0.9\columnwidth,draft=false]{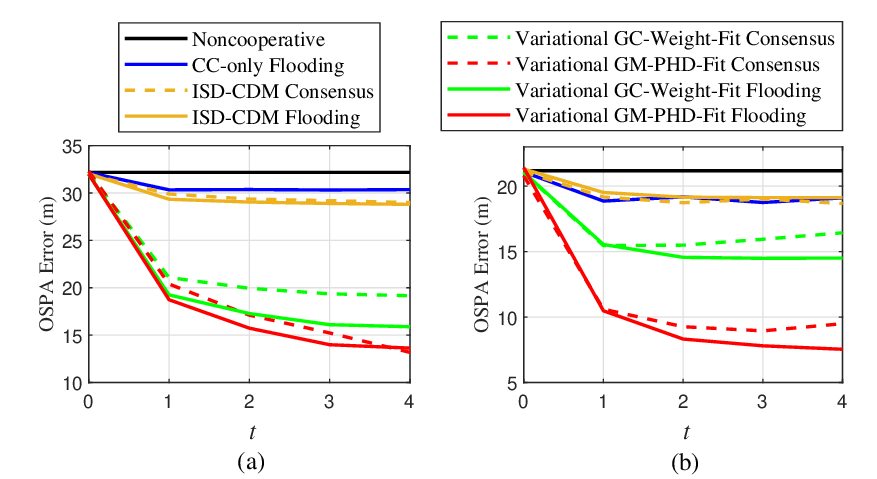}\\
  \vspace{-2mm}
  \caption{Average OSPA error of the heterogeneous fusion of MB and LMB filters. (a) result of the MB filters, (b) result of the LMB filters.} \label{fig:MBLMB}
    \vspace{-2mm}
\end{figure}

\begin{figure}
	\centering	\includegraphics[width=1.0\columnwidth,draft=false]{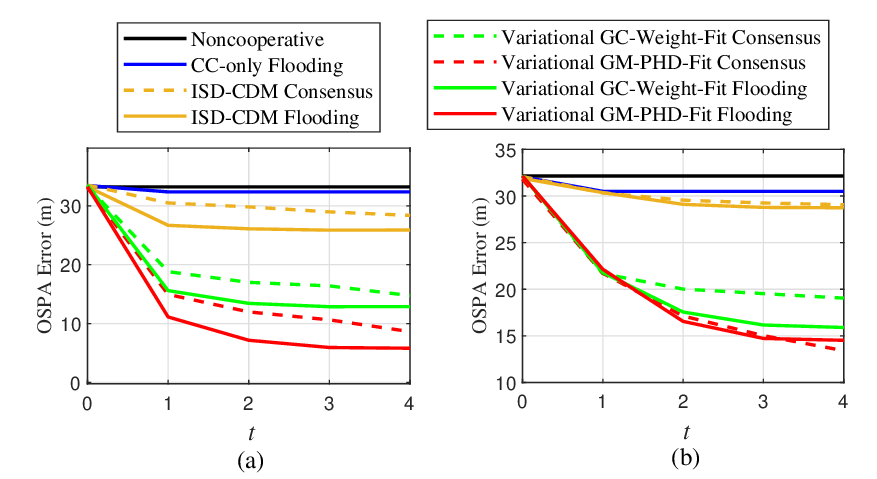}\\
 \vspace{-2mm}
	\caption{Average OSPA error of the heterogeneous fusion of PHD and MB filters. (a) result of the PHD filters, (b) result of the MB filters.} \label{fig:PHDMB}
	\vspace{-2mm}
\end{figure}

\begin{figure}
	\centering	\includegraphics[width=1.0\columnwidth,draft=false]{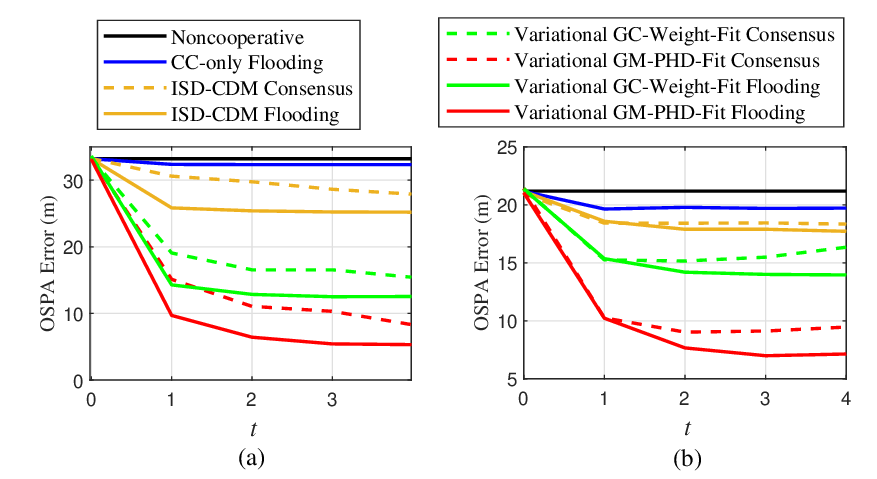}\\
 \vspace{-2mm}
	\caption{Average OSPA error of the heterogeneous fusion of PHD and LMB filters. (a) result of the PHD filters, (b) result of the LMB filters.} \label{fig:PHDLMB}
	\vspace{-2mm}
\end{figure}

\subsection{Heterogeneous PHD, MB and LMB Fusion}
We consider the heterogeneous fusion of two forms of filters. The average-over-time OSPA errors of different fusion approaches in the case of MB(6)-LMB(6), PHD(6)-MB(6) and PHD(6)-LMB(6) filter cooperation are given in Figs. \ref{fig:MBLMB}-\ref{fig:PHDLMB}, respectively, where the number in the bracket behind each filter indicates the number of sensors running that type of filter in the simulation. The results are averaged for each type of filters, separately.
Similar to the case of homogeneous fusion, the proposed variational GM-PHD-fit approach can benefit the PHD filters (maximally 85\%) more than for the MB/LMB filters (maximally 58\% and 69\%, respectively). This is because the PHD-AA fusion seeks the first moment consensus, which is simply insufficient for the the multi-object posterior (namely MB/LMB) consensus in the MB/LMB filters while what is calculated in the PHD filter is just the PHD over time. Similar with the homogeneous fusion, the fusion result becomes worse with the increase of $t$ after $t\geq 2$ in the case of the LMB filters, whether it cooperates with the PHD or MB filters.



\begin{figure*}
	\centering
	\includegraphics[width=15cm]{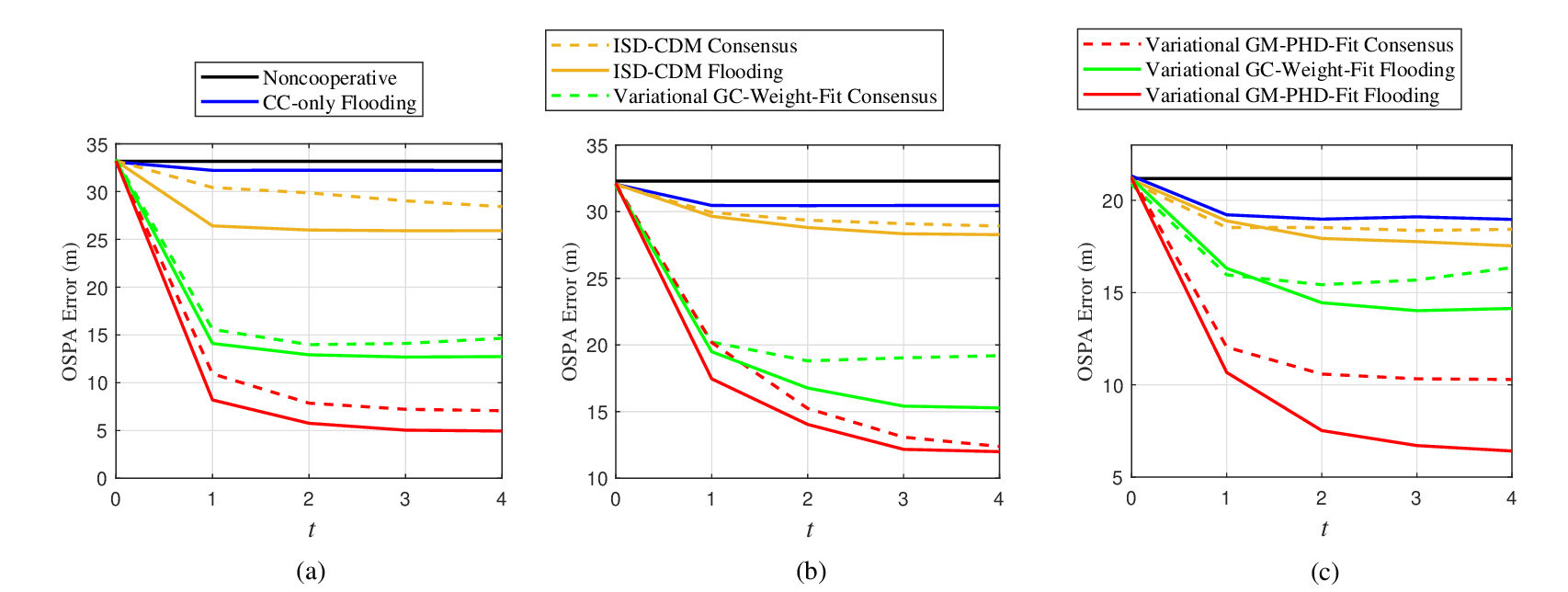}\\
 \vspace{-2mm}
	\caption{Average OSPA error of the heterogeneous fusion of PHD (a), MB (b) and LMB (c) filters.} \label{fig:PHDMBLMB}
	\vspace{-2mm}
\end{figure*}

Further on, we consider the heterogeneous fusion of the PHD(4)-MB(4)-LMB(4) filters. The average OSPA error of each type of filter is given in Fig. \ref{fig:PHDMBLMB}, which confirms that the heterogeneous unlabeled PHD-AA fusion can benefit all filters, where the reduction of OSPA error due to fusion can be as much as 85\%, 62\% and 72\%, for the PHD, MB and LMB filters, respectively. Fig. \ref{fig:hetercommunicationcost} shows the ACC in this scenario which complies with the case of homogeneous PHD filter fusion as shown in Fig. \ref{fig:PHDcommunicationcost}.

\begin{figure}
  \centering
  \includegraphics[width=7cm]{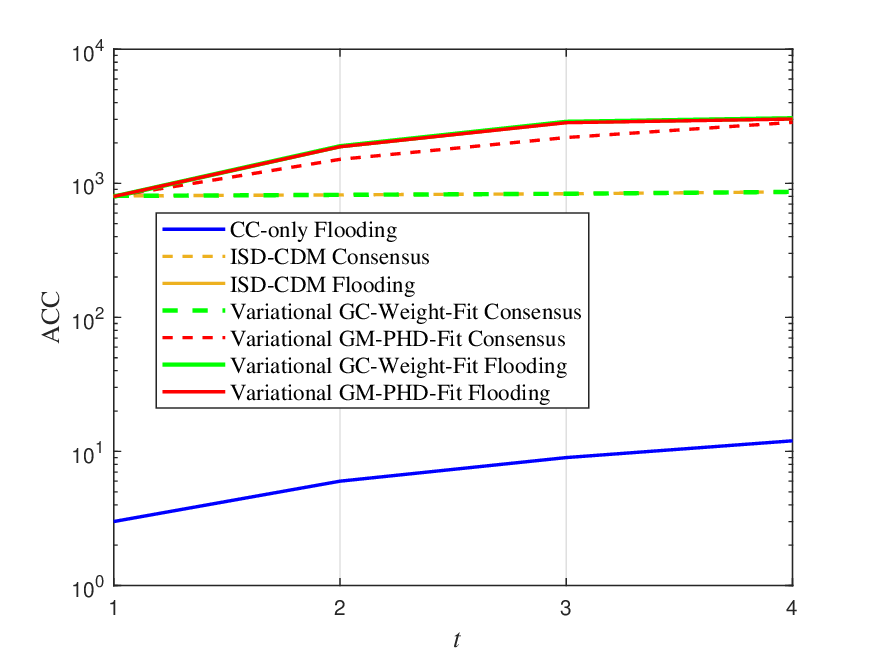}\\
  \vspace{-2mm}
  \caption{ACC of the heterogeneous fusion of PHD, MB and LMB filters. 
  } \label{fig:hetercommunicationcost}
  \vspace{-2mm}
\end{figure}

\subsection{Hybrid GM-PHD-fit} \label{sec:hyb_Simu}
Finally, we consider a higher level of heterogeneous fusion where the sensors run different types of RFS filters and different fusion approaches. This is best adapted for the case different sensors have different computing and communication capacities. 
To be specific, 4 nodes run the PHD filter and the ISD-CDM fusion approach to cooperate with the others, 4 nodes run the LMB filter and the variational GM-PHD-fit fusion approach and the rest 4 nodes run the MB filter and the variational GC-weight-fit fusion approach. The resulted average OSPA errors for each type of filters are shown in Fig. \ref{fig:Hybrid} (a) and (b) for the case of using consensus and flooding for inter-node communication, respectively. The results show that all three types of filters have been significantly improved in accuracy by their respective heterogeneous fusion approaches, which confirms the effectiveness of the proposed PHD-AA fusion approaches for hybrid, heterogeneous fusion. However, more significant divergence occurs in both the LMB (and even MB) filters with the increase of the number of fusion iterations in the consensus mode. The reasons can be twofold. First, the consensus algorithm using Metropolis weights as aforementioned does not guarantee better fusion results in with the increase of the number $t$ of consensus steps. Second, the ISD-CDM approach (using learning rate $\alpha_i =0.2$ and fading rate $\beta_i=0.6$ in this case) does not guarantee convergence and may diverge with the increase of $t$ as shown in \cite{Li23Heterogeneous}. 

\begin{figure}
	\centering
	\includegraphics[width=1.0\columnwidth,draft=false]{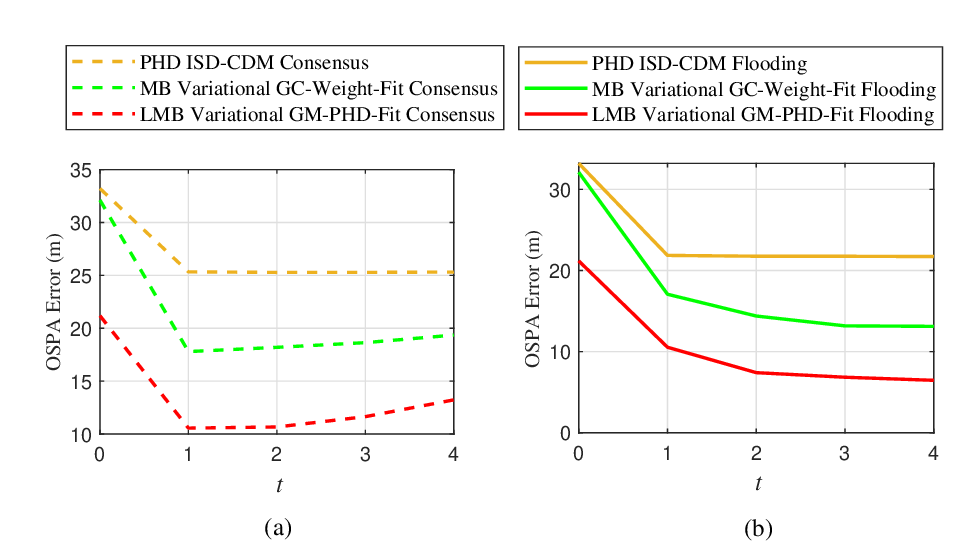}\\
 \vspace{-2mm}
	\caption{Average OSPA error of the hybrid heterogeneous fusion of PHD, MB and LMB filters in the case of (a) consensus  and (b) flooding, respectively. 
 } \label{fig:Hybrid}
	\vspace{-2mm}
\end{figure}

\section{Conclusion} \label{sec:conclusion}
We investigate the variational approximation approach to heterogeneous fusion of PHD, MB and labeled MB filters which averages their unlabeled PHDs based on the GM implementation while preserving their individual posterior forms. The approach has been implemented in two versions: one merely adjusts the weights of the local Gaussian components and is efficient in communication and computation. The other revises both the weights and the local Gaussian mean and covariance parameters in an alternating manner and yields better results while being more costly in both communication and fusion calculation. Both versions can be implemented using the consensus or flooding algorithm for inter-filter communication.
Simulations have demonstrated the effectiveness of our proposed approaches for both homogeneous and heterogeneous PHD-MB-LMB filter fusion. 
Improvement can be expected by performing proper track-to-track association prior to MB and LMB fusion. 

\appendix
\subsection{Proof of \eqref{eq:RFS-AA-Whole-KLD}} \label{sec:app-PHD-BFoM}
We define the following probability densities by normalizing the PHDs
    \begin{align}
    d^{\text{AA}}_{\mathcal{S}'_s}(\mathbf{x})  & \triangleq \frac{D^{\text{AA}}_{\mathcal{S}'_s}(\mathbf{x}) }{ \hat{N}^{\text{AA}}_{\mathcal{S}'_s}}, \label{eq:Def_d_AA} \\
    d_r(\mathbf{x})  & \triangleq \frac{D_r(\mathbf{x})}{ \hat{N}_r}, r\in \mathcal{S}'_s, \label{eq:Def_di} \\
    g'(\mathbf{x})  & \triangleq \frac{g(\mathbf{x})}{\hat{N}^{\text{AA}}_{\mathcal{S}'_s}}, \label{eq:Def_g'}
    \end{align}
where $\hat{N}^{\text{AA}}_{\mathcal{S}'_s}= \int_{\mathcal{X}} D^{\text{AA}}_{\mathcal{S}'_s}(\mathbf{x})d \mathbf{x}, \hat{N}_r = \int_{\mathcal{X}} D_r(\mathbf{x})d \mathbf{x}, r\in \mathcal{S}'_s$ and \eqref{eq:CC-constraint-BFoM} was used in \eqref{eq:Def_g'}.

As proven in \cite{Kulhavy96,Abbas09} and analyzed in \cite{Li24SomeResults}, one has
    \be
            d^{\text{AA}}_{\mathcal{S}'_s}(\mathbf{x})  = {\underset{g\in \mathcal{F}_{\mathcal{X}} }{\arg\min}} \sum\limits_{r \in {\mathcal{S}'_s}} w_r {\text{KL}}(d_r ||g ). \label{eq:d_AA_KLD_Min}
    \ee
We now expand the cost to be minimized in \eqref{eq:RFS-AA-Whole-KLD} as follows
\begin{align}
J(g) \triangleq & \sum\limits_{r \in {\mathcal{S}'_s}} w_r {\text{KL}}(D_r ||g ) \\
=& \sum\limits_{r \in {\mathcal{S}'_s}} w_r \int_{\mathcal{X}} \hat{N}_r d_r(\mathbf{x})\log \frac{ \hat{N}_r d_r(\mathbf{x})}{\hat{N}^{\text{AA}}_{\mathcal{S}'_s}g'(\mathbf{x})}d \mathbf{x} \label{eq:PHD-pdf-KLDcost} \\
=& \sum\limits_{r \in {\mathcal{S}'_s}} w_r \hat{N}_r \int_{\mathcal{X}} d_r(\mathbf{x})\log \frac{ \hat{N}_r }{\hat{N}^{\text{AA}}_{\mathcal{S}'_s}} d \mathbf{x} \nonumber \\
& + \sum\limits_{r \in {\mathcal{S}'_s}} w_r \hat{N}_r \int_{\mathcal{X}} d_r(\mathbf{x})\log \frac{d_r(\mathbf{x})}{g'(\mathbf{x})}d \mathbf{x} \nonumber \\
=& \sum\limits_{r \in {\mathcal{S}'_s}} w_r \hat{N}_r \log \frac{ \hat{N}_r }{\hat{N}^{\text{AA}}_{\mathcal{S}'_s}}  + \sum\limits_{r \in {\mathcal{S}'_s}} w_r {\text{KL}}(d_r ||g' ) \label{eq:J(p)-expansion}
\end{align}
where \eqref{eq:Def_di} and \eqref{eq:Def_g'} were used in \eqref{eq:PHD-pdf-KLDcost}.

It is evident that $\sum\limits_{r \in {\mathcal{S}'_s}} w_r \hat{N}_r  \log \frac{ \hat{N}_r }{\hat{N}^{\text{AA}}_{\mathcal{S}'_s}} $ is constant and the rest part of \eqref{eq:J(p)-expansion}, as well as $J(g)$, is minimized by $g'(\mathbf{x}) = d_{\text{AA}}(\mathbf{x})$ according to \eqref{eq:d_AA_KLD_Min}. This 
proves \eqref{eq:RFS-AA-Whole-KLD}.

\subsection{Proof of Lemma \ref{varlemma}} \label{sec:VAconvergence}
We start from problem 1 and denote $F_{1}(\Theta) \triangleq \mathrm{KL}(\phi||\varphi) $ and $F_{2}(\Theta,Q) \triangleq \sum_{a,b} \phi_{s,b|a}\  \mathrm{KL}(\mathcal{N}_a||\mathcal{N}_{s,b})$, and further,
\begin{align}
	F(\Theta,Q) & \triangleq F_{1}(\Theta) + F_{2}(\Theta,Q) \nonumber \\
	&=\mathrm{KL}(\phi||\varphi) + \sum_{s,a,b} \phi_{b|a}\  \mathrm{KL}(\mathcal{N}_a||\mathcal{N}_{s,b}), 
\end{align}
where $\Theta$ denotes the set of two variational parameters $\phi$ and $\varphi$, $Q$ denotes the parameter set of $\mathcal{N}_{s,b}$.

Note that both $F_{1}(\Theta)$ and $F_{2}(\Theta,Q)$ are the linear combinations of KLD, thus are non-negative. The goal of our proposed variational GM-PHD-fit approach is to minimize $F(\Theta,Q)$ by adjusting $\phi_{s,b|a}$, $\varphi_{s,a|b}$ and $\{\bm{\mu}_b,\bm{\Sigma}_b\}$ in an alternating manner. Starting from any feasible point $F(\Theta^0,Q^0)$, we firstly update the variational parameters to minimize the objective function, which results in
\begin{align}
	F(\Theta^1,Q^0)\le F(\Theta^0,Q^0), 
\end{align}
where $\Theta^1$ is the resulted new variational parameters. Then, we optimize the distribution parameters based on $\Theta^1$, yielding
\begin{align}
	F(\Theta^1,Q^1)\le F(\Theta^1,Q^0) \ist. 
\end{align}

Based on $F(\Theta^1,Q^1)$, the above alternating optimization of the variational parameters and the L-GC parameters will be iterated till get the desired results. By this, the objective function will be monotonically decreasing, namely, optimization convergence.

\bibliographystyle{IEEEtran}
\bibliography{VBHeterogenousFusion}

\end{document}